\documentclass[journal]{IEEEtran}
\usepackage{cite}

\ifCLASSINFOpdf
\else
\fi
\usepackage{tikz}

\usepackage[cmex10]{amsmath}
\usepackage{amssymb}
\usepackage{amsthm}

\newtheorem{theorem}{Theorem}[section]
\newtheorem{proposition}[theorem]{Proposition}
\newtheorem{lemma}[theorem]{Lemma}
\newtheorem{corollary}[theorem]{Corollary}

\renewcommand{\epsilon}{\varepsilon}

\newcommand{\vect}[1]{\mathbf{#1}} 

\newcommand{\bff}{\vect{f}}

\newcommand{\bfx}{\vect{x}}
\newcommand{\bfy}{\vect{y}}

\newcommand{\bfX}{\vect{X}}
\newcommand{\bfY}{\vect{Y}}

\newcommand{\Reals}{\mathbb R} 
\newcommand{\reals}{\mathbb R} 
\newcommand{\naturals}{\mathbb N} 
\newcommand{\supp}{\textnormal{supp}}
\newcommand{\abs}[1]{\lvert#1\rvert} 
\newcommand{\card}[1]{\abs{#1}} 
\newcommand{\bigcard}[1]{\bigl\lvert#1\bigr\rvert}

\newcommand{\comp}[1]{{#1}^{\textnormal{c}}} 
\newcommand{\setcomp}[1]{{#1}^{\textnormal{c}}} 
\newcommand{\set}[1]{\mathcal{#1}} 
\newcommand{\E}[1]{\operatorname{E}[#1]} 
\newcommand{\bigE}[1]{\operatorname{E}\bigl[#1\bigr]} 
\newcommand{\BiggE}[1]{\operatorname{E}\Biggl[#1\Biggr]} 

\newcommand{\Cal}{C_{\ell}(\rho)}
\newcommand{\Calone}{C_{\ell}(1)}
\newcommand{\Calf}{C_{\textnormal{$\ell$,fb}}(\rho)}
\newcommand{\Coff}{R_{\textnormal{cutoff}}(\rho)}
\newcommand{\Cofff}{R_{\textnormal{cutoff,fb}}(\rho)}
\newcommand{\Czero}{C_0}
\newcommand{\Czerof}{C_{0,\textnormal{fb}}}
\newcommand{\Ceo}{C_{\textnormal{0-u}}}
\newcommand{\Ceof}{C_{\textnormal{0-u,fb}}}

\begin{document}
%
\title{On the Listsize Capacity with Feedback}
%
%
%

\author{Christoph~Bunte,
        Amos~Lapidoth,~\IEEEmembership{Fellow,~IEEE}
\thanks{C. Bunte and A. Lapidoth are with the Signal and Information Processing
Laboratory at ETH Zurich. E-mail: \{bunte,lapidoth@isi.ee.ethz.ch\}.}
\thanks{This paper was presented 
in part at the 2013 IEEE Information Theory Workshop (ITW) in Seville, Spain.}
}

\maketitle

\begin{abstract}
The listsize capacity of a discrete memoryless channel is the largest
transmission rate for which the 
expectation---or, more generally, the $\rho$-th moment---of the number of messages 
that could have produced the output of the channel approaches one as the
blocklength tends to infinity. We show that for channels with feedback this rate
is upper-bounded by the maximum of Gallager's $E_0$ function divided by $\rho$, 
and that equality holds when the zero-error capacity of the channel is positive.
To establish this inequality we prove that feedback does not increase the cutoff
rate. Relationships to other notions of channel capacity are explored. 
\end{abstract}


%
\IEEEpeerreviewmaketitle

\section{Introduction and Results}
The main focus of this paper is the \emph{listsize capacity} of discrete memoryless channels
(DMCs) with feedback. We begin by recalling the definition of the listsize
capacity and, to put things into perspective, some other notions of channel capacity. 
\subsection{Various Notions of Capacity}
If a code for a DMC is to be decoded without
errors, then for every sequence of output letters there can be at
most one message that, when fed to the encoder, can produce it. 
The \emph{zero-error capacity} $\Czero$ of a DMC is the largest rate of
codes with this property. 
Determining $\Czero$ for arbitrary DMCs is one of the longest standing open problems in
Information Theory~\cite{korner1998zero}. 
If we only require that the correct message be decodable 
with probability approaching one as the blocklength tends to infinity,
then suddenly the problem becomes tractable.
Indeed, the largest rate achievable in this sense is the \emph{Shannon capacity} $C$. 

The zero-error capacity is a purely combinatoric quantity: it depends only on
the zeros of the channel matrix. The Shannon capacity, in contrast, is a continuous function of the channel matrix. 
Two notions of channel capacity that lie between these two extremes are the
listsize capacity and the \emph{zero-undetected-error~(z.u.e.) capacity}; 
they may be defined as follows. 
\begin{enumerate}
\item
Consider a decoder that outputs the list of all the messages that could
have produced the given output of the channel. 
The listsize capacity is the largest rate achievable in the sense that 
the $\rho$-th moment of the length of this list approaches one as the
blocklength tends to infinity~\cite{ahlswede1996erasure,telatar1997zero}. 
It is denoted by $\Cal$. In this paper, $\rho$ can be any number greater than zero.
\item
Consider a decoder that either outputs the correct message
(when there is a unique message that could have produced the given output) or
declares an erasure (otherwise). The z.u.e.\ capacity is the largest rate achievable 
in the sense that the probability of erasure approaches zero
as the blocklength tends to
infinity~\cite{csiszar1995channel,ahlswede1996erasure}. It is denoted by $\Ceo$. 
\end{enumerate}

For any given channel, 
\begin{equation}
\label{eq:inequalities}
\Czero \leq \Cal \leq \Ceo \leq C.
\end{equation}
The first and third inequalities are obvious, and the second inequality is proved in Proposition~\ref{prop:list_erasure_relation} ahead.
The listsize and z.u.e.\ capacities are not purely combinatoric quantities, nor are they
continuous functions of the channel matrix. 
But like $\Czero$, determining $\Cal$ or $\Ceo$ for arbitrary DMCs is, to the
best of our knowledge, an open problem. 
\subsection{Feedback and New Results}
The picture changes when there is a noiseless feedback link from the output of the channel to the encoder. 
Indeed, for channels with feedback, the zero-error capacity was proved by
Shannon~\cite{shannon1956zero} to be
equal to the single-letter expression~\eqref{eq:shannon_zero} ahead. The z.u.e.\
capacity with feedback was found in~\cite{nakiboglu2012errors,bunte2012zero} and
can be expressed as in~\eqref{eq:nakiboglu} ahead. 

Encouraged by these results, we focus here on the listsize capacity with feedback
$\Calf$. For channels with positive zero-error capacity we prove that
$\Calf$ equals the maximum over all input distributions of the ratio of
Gallager's $E_0$ function (\cite[p.\ 138]{gallager1968information} or~\eqref{eq:gallager_E0} in
Appendix~\ref{appendix:gallager}) to $\rho$. Moreover, this maximum is always an
upper bound on~$\Calf$: 
\begin{theorem}
\label{thm:list_cutoff_achieve}
For any $\rho>0$,
\begin{equation}
\label{eq:list_cutoff_achieve}
\Calf \leq  \max_{P} \frac{E_0(\rho,P)}{\rho},
\end{equation}
with equality if $\Czero>0$. 
\end{theorem}
\begin{figure}
\centering
\begin{tikzpicture}[scale=2, >=latex, thick]
\draw[-] (0,0) node[left] {$2$} -- (2,0) node[right] {$2$};
\draw[-] (0,1) node[left] {$1$} -- node[below] {$1-\epsilon$} (2,1) node[right] {$1$};
\draw[-] (0,2) node[left] {$0$} -- node[above] {$1-\epsilon$} (2,2) node[right] {$0$};
\draw[-] (0,1) -- node[above, near start] {$\epsilon$} (2,2);
\draw[-] (0,2) -- node[above, near start] {$\epsilon$} (2,1);
\end{tikzpicture}
\caption{A channel with $\Calf>\Cal$.}
\label{fig:bsc_plus_wire}
\end{figure}
A lower bound on~$\Calf$ when~$\Czero=0$ is provided in
Section~\ref{sec:feedback_lb} (Theorem~\ref{thm:feedback_lb}).
We can use Theorem~\ref{thm:list_cutoff_achieve} to show: 
\begin{proposition}
\label{prop:feedback_can_increase}
Irrespective of $\rho>0$, feedback can increase the listsize capacity.
\end{proposition}
\begin{proof}
The channel in Figure~\ref{fig:bsc_plus_wire} has positive 
zero-error capacity, and $\max_P E_0(\rho,P)/\rho$ approaches~$\log 3$ as~$\epsilon$ 
tends to zero. Consequently, by Theorem~\ref{thm:list_cutoff_achieve}, $\Calf$
approaches $\log 3$ as $\epsilon$ tends to zero. But according to
Proposition~\ref{prop:combine} ahead we may combine the output symbols 0 and~1 without
altering~$\Cal$, so~$\Cal \leq \log 2$ because the resulting output alphabet is binary.
\end{proof}
We note that also $\Ceof>\Ceo$ for the channel in
Figure~\ref{fig:bsc_plus_wire}~\cite{bunte2012zero}. 

The direct part of Theorem~\ref{thm:list_cutoff_achieve} is proved in
Section~\ref{sec:cutoff_achieve}, where we also show that the
inequality~\eqref{eq:list_cutoff_achieve} need not be tight if $\Czero=0$. 
In order to derive~\eqref{eq:list_cutoff_achieve}, we recall~\cite{telatar1997zero} the following
operational meaning of the right-hand side of~\eqref{eq:list_cutoff_achieve}.
Consider the list of all the messages that under a uniform prior are at least as likely as the correct one 
given the output of the channel. The \emph{cutoff rate $\Coff$} 
is the largest rate of codes for which the $\rho$-th moment 
of the length of this list approaches one as the blocklength tends to
infinity. Since the list of messages that could have produced the output
contains those that are at least as likely as the correct one, 
\begin{equation}
\label{eq:cal_leq_coff}
\Cal \leq \Coff,
\end{equation}
and, for channels with feedback,
\begin{equation}
\label{eq:229_inequ}
\Calf \leq \Cofff.
\end{equation}
En route to the converse part of
Theorem~\ref{thm:list_cutoff_achieve} we prove:
\begin{theorem}
\label{thm:cutoff_no_increase}
For any $\rho>0$,
\begin{equation}
\label{eq:253_534}
\Cofff = \max_{P} \frac{E_0(\rho,P)}{\rho}.
\end{equation}
\end{theorem}
Inequality~\eqref{eq:list_cutoff_achieve} follows directly
from~\eqref{eq:229_inequ} and Theorem~\ref{thm:cutoff_no_increase}.
The converse part of Theorem~\ref{thm:cutoff_no_increase} is proved in
Section~\ref{sec:cutoff_fb}.
The achievability part follows from the well-known 
result (e.g.,~\cite{telatar1997zero})
\begin{equation}
\label{eq:225_gallager}
\Coff =  \max_{P} \frac{E_0(\rho,P)}{\rho}
\end{equation}
combined with the trivial fact that $\Cofff \geq \Coff$. 
To keep this paper self-contained, we prove the achievability part
of~\eqref{eq:225_gallager} in Appendix~\ref{appendix:cutoff}.\footnote{The case
where $\rho=1$ follows essentially from Gallager's derivation of the random
coding error exponent~\cite[Sec.\ 5.6]{gallager1968information}. The general case, however,
requires a bit more work.}

As a corollary to Theorem~\ref{thm:cutoff_no_increase}, we obtain that feedback
does not increase the cutoff-rate:
\begin{corollary}
\label{cor:cutoff_feedback_no_increase}
For any $\rho>0$, 
\begin{equation}
\label{eq:cutoff_E0}
\Cofff=\Coff.
\end{equation}
\end{corollary}

This paper also contains the following other contributions:
\begin{enumerate}
\item A generalization of Forney's~\cite{forney1968exponential} lower bound on $\Calone$ to $\Cal$ for all
$\rho>0$ and a proof that the $n$-letter version of this bound becomes tight as
$n\to\infty$ even when the input distributions (PMFs) are restricted to be uniform over their
support; see Proposition~\ref{prop:forney_cal_rho} and Section~\ref{sec:forney}. 
\item Sufficient conditions for equality in $\Cal\leq \Coff$; see
Proposition~\ref{prop:factorize}.
\item A simple method to tighten the upper bounds in $\Cal \leq \Coff$ and $\Ceo \leq C$;
see Proposition~\ref{prop:combine}.
\item A proof that $\lim_{\rho\to0} \Cal = \Ceo$; see
Proposition~\ref{prop:list_erasure_relation}.
\item The limit of $\Calf$ as $\epsilon\to0$ for a class of ``$\epsilon$-noise'' channels; see Proposition~\ref{prop:low_noise_fb}. 
\end{enumerate}

\subsection{Notation and Definitions}
The cardinality of a finite set~$\set{X}$ is denoted by~$\card{\set{X}}$. 
We use boldface letters to denote $n$-tuples, e.g., $\bfx=(x_1,\ldots,x_n)$, and
uppercase boldface letters for random $n$-tuples, e.g., $\bfY=(Y_1,\ldots,Y_n)$. 
Sometimes we use~$x^i$ as shorthand for $(x_1,\ldots, x_i)$ when $0\leq i \leq
n$, where $x^0$ is the empty tuple. All logarithms are natural logarithms. We adopt the 
convention that~$a\log (b/c)$ equals zero if $a=0$; equals $+\infty$ if
$a>0$, $b>0$, and $c=0$; and equals $-\infty$ if $a>0$, $b=0$, and $c>0$. 
For information-theoretic quantities like entropy and relative entropy we follow the
notation in~\cite{csiszar2011information}. In some of the proofs we use basic
results about types, all of which can be found
in~\cite[Chapter~2]{csiszar2011information}. 
In particular, the set of all sequences of type $P$ is denoted by $T_P$. The set of all
sequences whose conditional type is $V$ given $\bfx$, i.e., the $V$-shell of $\bfx$, is
denoted by $T_V(\bfx)$. 
Throughout~$(\delta_n)_{n\geq1}$ is used to denote sequences of nonnegative numbers that
tend to zero. We write $\delta_n'$, $\delta_n''$, etc., if we want to emphasize that different such
sequences are being used. The indicator function is denoted by~$1\{\cdot\}$.

A \emph{discrete memoryless channel (DMC)} is specified by its transition law (channel
matrix) $W(y|x)$, $x\in \set{X}$, $y\in\set{Y}$, where~$\set{X}$ and~$\set{Y}$ are finite
input and output alphabets. 
If~$P$ is a probability mass function (PMF) on~$\set{X}$, then $PW$ denotes the
distribution induced on~$\set{Y}$ by $P$ and the transition law $W$
\begin{equation}
(PW)(y) = \sum_{x\in\set{X}} P(x)W(y|x),\quad y\in\set{Y}.
\end{equation}
We write $P^n$ for the product PMF on $\set{X}^n$ 
\begin{equation}
P^n(\bfx)=\prod_{i=1}^n P(x_i),\quad \bfx\in\set{X}^n.
\end{equation}
The support of a PMF $P$ is denoted by $\supp(P)$, i.e.,
$\supp(P)=\{x\in\set{X}:P(x)>0\}$. 
If $\set{A}\subseteq \set{X}$, we write~$P(\set{A})$ instead of
$\sum_{x\in\set{A}} P(x)$. Similarly, if $\set{B} \subseteq \set{Y}$, we write
$W(\set{B}|x)$ instead of $\sum_{y\in \set{B}} W(y|x)$. 

In the absence of feedback, a blocklength-$n$ rate-$R$ encoder is a
mapping\footnote{More precise would be the integer part of $e^{nR}$, but for
typographical reasons we write $e^{nR}$ instead of $\lfloor e^{nR} \rfloor$.}
\begin{equation}
\label{eq:encoder}
f\colon \{1,\ldots,e^{nR}\} \to \set{X}^n.
\end{equation}
The domain of $f$ is the \emph{message set} and the (not necessarily distinct)
\emph{codewords} $f(1),\ldots,f(e^{nR})$ constitute the \emph{codebook}. 
We sometimes write $\bfx_m$ instead of $f(m)$ for the codeword to
which the encoder maps the $m$-th message.
Sending the $m$-th message induces on $\set{Y}^n$ the distribution 
\begin{equation}
W^n\bigl(\bfy\big|f(m)\bigr), \quad \bfy \in \set{Y}^n,
\end{equation}
where 
\begin{equation}
W^n(\bfy|\bfx) = \prod_{i=1}^n W(y_i|x_i),\quad \bfx \in \set{X}^n,\, \bfy \in
\set{Y}^n.
\end{equation}
We often use the notation
\begin{equation}
\set{X}(y) = \{x\in\set{X}: W(y|x)>0\},
\end{equation}
and
\begin{equation}
\set{X}^n(\bfy) = \{\bfx\in\set{X}^n:W^n(\bfy|\bfx)>0\}.
\end{equation}
Given an encoder $f$ as in~\eqref{eq:encoder}, we define the lists\footnote{We
use the word ``list'' in the sense of a set.} 
\begin{equation}
\set{L}(\bfy) = \bigl\{ m: W^n\bigl(\bfy\big|f(m)\bigr)>0\bigr\},\quad \bfy \in
\set{Y}^n,
\end{equation}
and
\begin{equation}
\set{L}(m,\bfy) = \bigl\{\tilde{m}: W^n\bigl(\bfy\big|f(\tilde{m})\bigr)\geq
W^n\bigl(\bfy\big|f(m)\bigr)\bigr\}.
\end{equation}
Stated differently,~$\set{L}(\bfy)$ is the list of all messages that can 
produce the output sequence $\bfy$, and~$\set{L}(m,\bfy)$ is the list of all
messages that under the uniform prior are at least as likely as the~$m$-th message 
given that~$\bfy$ is observed at the output. 

We can now give precise definitions of $\Cal$, $\Ceo$, and~$\Coff$. 
\begin{enumerate}
\item $\Cal$ is the supremum of all rates $R$ for which there exists a sequence of
blocklength-$n$ rate-$R$ encoders~$(f_n)_{n\geq 1}$ such that
\begin{equation}
\label{eq:cal_def}
\lim_{n\to\infty} e^{-nR} \sum_{m=1}^{e^{nR}} \sum_{\bfy \in \set{Y}^n}
W^n\bigl(\bfy\big|f_n(m)\bigr)\,
\card{\set{L}(\bfy)}^\rho =1.
\end{equation}
\item $\Ceo$ is the supremum of all rates $R$ for which there exists a sequence 
of blocklength-$n$ rate-$R$ encoders~$(f_n)_{n\geq 1}$ such that
\begin{equation}
\label{eq:ceo_def}
\lim_{n\to\infty} e^{-nR} \sum_{m=1}^{e^{nR}} \sum_{\bfy:
\card{\set{L}(\bfy)}\geq 2} W^n\bigl(\bfy\big|f_n(m)\bigr) = 0.
\end{equation}
\item $\Coff$ is the supremum of all rates $R$ for which there exists a sequence 
of blocklength-$n$ rate-$R$ encoders~$(f_n)_{n\geq 1}$ such that
\begin{equation}
\label{eq:coff_def}
\lim_{n\to\infty} e^{-nR} \sum_{m=1}^{e^{nR}} \sum_{\bfy\in\set{Y}^n}
W^n\bigl(\bfy\big|f_n(m)\bigr)\,\card{\set{L}(m,\bfy)}^\rho = 1.
\end{equation}
\end{enumerate}
It follows from Gallager's derivation of the Channel Coding
Theorem~\cite[Ch.\ 5]{gallager1968information} that the Shannon capacity $C$ can be achieved by
strict ML-decoding, i.e., by a decoder that either produces the unique message of
maximum likelihood (if there is one) or erases
(otherwise). Consequently, we may define~$C$ in terms of 
the list $\set{L}(m,\bfy)$ as follows.
\begin{enumerate}
\item[4.] $C$ is the supremum of all rates $R$ for which there exists a sequence
of blocklength-$n$ rate-$R$ encoders~$(f_n)_{n\geq 1}$ such that
\begin{equation}
\label{eq:c_def}
\lim_{n\to\infty} e^{-nR} \sum_{m=1}^{e^{nR}} \sum_{\bfy:
\card{\set{L}(m,\bfy)}\geq 2} W^n\bigl(\bfy\big|f_n(m)\bigr) = 0.
\end{equation}
\end{enumerate}
(The above definitions remain unchanged when the average over the
messages is replaced with the maximum. This follows from a standard expurgation
argument.)

To extend the above definitions to channels with feedback, we replace $f$ with an $n$-tuple
$(f^{(1)},\ldots,f^{(n)})$, where 
\begin{equation}
\label{eq:feedback_encoder}
f^{(i)} \colon \{1,\ldots,e^{nR}\} \times \set{Y}^{i-1} \to \set{X},\quad
i=1,\ldots,n.
\end{equation}
(By convention, $\set{Y}^0$ contains only the empty tuple.)
In this case, sending the $m$-th message induces on $\set{Y}^n$ the distribution
\begin{equation}
\prod_{i=1}^n W\bigl(y_i\big|f^{(i)}(m, y^{i-1})\bigr),\quad \bfy\in\set{Y}^n.
\end{equation}
The definitions of $\set{L}(\bfy)$, $\set{L}(m,\bfy)$, $\Calf$, $\Ceof$, and $\Cofff$ are
analogous to their no-feedback counterparts. 

\subsection{Bounds---Old and New}
We begin with some known lower bounds on $\Cal$ and~$\Ceo$. 
Forney~\cite{forney1968exponential} showed that
\begin{equation}
\label{eq:forney_ceo}
\Ceo \geq \max_{P} - \sum_{y\in\set{Y}} (PW)(y) \log P\bigl(\set{X}(y)\bigr)
\end{equation}
and
\begin{equation}
\label{eq:forney_cal}
\Calone \geq \max_{P} -\log \sum_{y\in\set{Y}} (PW)(y) P\bigl(\set{X}(y)\bigr),
\end{equation}
where the maxima are over all PMFs on $\set{X}$.
Forney's bounds can be derived using
standard random coding where each component of each codeword is drawn
independently according to a PMF~$P$. 
In Section~\ref{sec:forney} we prove the following generalization
of~\eqref{eq:forney_cal} (also using standard random coding).
\begin{proposition}
\label{prop:forney_cal_rho}
For any $\rho>0$,
\begin{equation}
\label{eq:forney_cal_rho}
\Cal \geq \max_P -\rho^{-1} \log \sum_{y\in\set{Y}} (PW)(y)
P\bigl(\set{X}(y)\bigr)^\rho.
\end{equation}
\end{proposition}
Neither~\eqref{eq:forney_ceo} nor~\eqref{eq:forney_cal_rho} is tight 
in general.\footnote{An example where they are not tight is the
Z-channel; see~\cite[Example~4.1]{telatar1992phd}.}
Tighter bounds can be
derived using random coding over constant composition 
codes~\cite{telatar1992phd, ahlswede1996erasure, csiszar1995channel, telatar1997zero}: The
corresponding bound on~$\Ceo$ is 
\begin{equation}
\label{eq:const_comp_lb_ceo}
\Ceo \geq \max_{P} \min_{\substack{V\ll W\\PV=PW}} I(P,V),
\end{equation}
where the minimization is over all auxiliary channels $V(y|x)$, $x\in \set{X}$,
$y\in \set{Y}$, such that $V(y|x) = 0$ whenever $W(y|x)=0$ (i.e., $V\ll W$) and
such that the induced output distribution under $P$ is the same as under
the true channel $W$ (i.e., $PV=PW$). The corresponding bound on $\Cal$ is 
\begin{equation}
\label{eq:const_comp_lb_cal}
\Cal \geq \max_{P} \min_{\substack{V,V'\\V\ll W\\PV=PV'}} I(P,V)+
\rho^{-1} D(V'||W|P).
\end{equation}
It is shown in~\cite{ahlswede1996erasure} and \cite{telatar1992phd} 
that~\eqref{eq:const_comp_lb_ceo} is
at least as tight as~\eqref{eq:forney_ceo}. Appendix~\ref{appendix:D}
contains a proof that~\eqref{eq:const_comp_lb_cal} is
at least as tight as~\eqref{eq:forney_cal_rho}. (This result may not have
appeared in print before.) However, the weaker bounds are
simpler because no minimization over auxiliary channels is required.

We can tighten any of the above lower bounds by applying them to the 
channel~$W^n(\bfy|\bfx)$, $\bfx\in\set{X}^n$, $\bfy\in\set{Y}^n$ and normalizing the result by~$1/n$. 
Indeed, any blocklength-$\nu$ \mbox{rate-$R$} code for the channel $W^n$ is a
\mbox{blocklength-$n \nu$} \mbox{rate-$R/n$} code for the 
channel $W$.\footnote{For blocklengths that are not divisible by $n$, we can 
interpolate as follows. Suppose the blocklength is $n\nu + \ell$ 
where $1\leq \ell < n$. Then we use a good code for~$W^n$ of blocklength~$\nu$ 
and rate $R$, and we extend it to a blocklength-$(n\nu+\ell)$ code for $W$ 
by padding $\ell$ dummy symbols. 
Accordingly, the last~$\ell$ output symbols  are ignored at the receiver. 
The rate of the resulting code for~$W$ is $\nu R/(n \nu +\ell)$,
and this approaches $R/n$ as $\nu\to\infty$.}

To give a concrete example, the $n$-letter version of~\eqref{eq:forney_cal_rho} is
\begin{equation}
\label{eq:forney_cal_rho_multi}
\Cal \geq \frac{1}{n}\max_P -\rho^{-1} \log \sum_{\bfy \in \set{Y}^n}
(PW^n)(\bfy) P\bigl(\set{X}^n(\bfy)\bigr)^\rho,
\end{equation}
where the maximum is over all PMFs on $\set{X}^n$. 
A numerical evaluation in~\cite{ahlswede1996erasure} of the one and two-letter versions
of~\eqref{eq:const_comp_lb_ceo} for a specific channel suggests that a
strict improvement is possible, and thus that~\eqref{eq:const_comp_lb_ceo} is
not always tight.
In~\cite{ahlswede1996erasure} and~\cite{telatar1997zero} it is shown that the
$n$-letter versions of~\eqref{eq:const_comp_lb_ceo}
and~\eqref{eq:const_comp_lb_cal} become tight as $n\to\infty$. 
In~\cite{bunte2014zero} it is shown that also the $n$-letter version of the weaker
bound~\eqref{eq:forney_ceo} becomes tight as $n\to\infty$, and that this is true
even when the input PMFs are restricted to be uniform over their support. 
In Section~\ref{sec:forney} we prove a similar statement for the $n$-letter
version of~\eqref{eq:forney_cal_rho}. 

The aforementioned limits are not computable in
general, but they can be useful nonetheless. For example,
in~\cite{bunte2014zero} the multiletter version of~\eqref{eq:forney_ceo} 
is used to derive an upper bound on $\Ceo$ for the
class of $\epsilon$-noise channels~(see below).

We now discuss upper bounds on $\Cal$ and $\Ceo$. Specifically, recall~\eqref{eq:cal_leq_coff} and
the rightmost inequality in~\eqref{eq:inequalities}:
\begin{equation}
\label{eq:713_bounds}
\Ceo \leq C\quad \text{and}\quad \Cal \leq \Coff.
\end{equation}
For a large class of channels the bounds in~\eqref{eq:713_bounds}
are tight:
\begin{proposition}
\label{prop:factorize}
The inequalities in~\eqref{eq:713_bounds} hold with equality if there exist functions
$A\colon\set{X}\to(0,\infty)$ and $B\colon\set{Y}\to(0,\infty)$ such that
\begin{equation}
\label{eq:factorize}
W(y|x) = A(x)B(y),\quad\text{if $W(y|x)>0$.}
\end{equation}
\end{proposition}
\begin{proof}
The hypothesis implies that the lists $\set{L}(\bfy)$ and $\set{L}(m,\bfy)$ coincide whenever
$W^n(\bfy|f(m))>0$ and constant composition codes are used.\footnote{Constant
composition codes comprise codewords of the same type~\cite[p.\ 144]{csiszar2011information}.} 
Indeed, observe that if $\bfx$ and $\bfx'$ are codewords
of the same type, then
\begin{align}
W^n(\bfy|\bfx) &= \bigg(\prod_{i=1}^n A(x_i)\biggr) \biggl(\prod_{j=1}^n
B(y_j)\biggr)\notag\\
&=\biggl(\prod_{i=1}^n A(x'_i)\biggr)\biggl(\prod_{j=1}^n B(y_j)\biggr)\notag\\
&=W^n(\bfy|\bfx'),
\end{align}
where the first and last equality hold 
provided that $W^n(\bfy|\bfx)>0$ and $W^n(\bfy|\bfx')>0$. Thus, all codewords
with positive likelihood have the same likelihood. 

Since every code has a constant composition subcode of exponentially the same
size (there are only polynomially many types), 
the proposition follows by comparing~\eqref{eq:cal_def} and~\eqref{eq:coff_def},
and~\eqref{eq:ceo_def} and~\eqref{eq:c_def}. 
\end{proof}
Proposition~\ref{prop:factorize} is essentially due to Csisz\'ar and
Narayan~\cite{csiszar1995channel} (they considered only $\Ceo$), who also observed that all channels with acyclic
channel graphs\footnote{The channel graph of a DMC $W$ is the undirected bipartite graph whose two
independent sets are $\set{X}$ and $\set{Y}$, and where there is an edge between $x$
and $y$ if $W(y|x)>0$. It is customary to draw the inputs on the left and the
outputs on the right, and to label the edges with the transition probabilities.
Acyclic means that we cannot find distinct inputs $x_1,\ldots,x_n$ and distinct
outputs $y_1,\ldots,y_n$ such that $W(y_i|x_i)>0$ and $W(y_{i}|x_{i+1})>0$ for all
$i\in \{1,\ldots,n\}$ where $n\geq 2$ and $x_{n+1} = x_1$.} can be factorized as in~\eqref{eq:factorize}. 
This important special case had been proved earlier by Pinsker and
Sheverdyaev~\cite{pinsker1970transmission} for $\Ceo$, and by
Telatar~\cite{telatar1997zero} for $\Cal$. (An intermediate result was
obtained by Telatar in~\cite[Sec.\ 4.3]{telatar1992phd}.) Notable examples of channels
with acyclic channel graphs are the Z-channel and the binary erasure channel. 
In~\cite{csiszar1995channel} it is conjectured that a necessary
condition for $\Ceo=C$ is that a factorization of the channel law in the sense
of~\eqref{eq:factorize} hold on
some capacity-achieving subset of inputs (which is clearly also sufficient). 


We can sometimes tighten the bounds in~\eqref{eq:713_bounds} by judiciously combining output
symbols: 
\begin{proposition}
\label{prop:combine}
If $y,y' \in \set{Y}$ are such that for every $x\in \set{X}$, $W(y|x)>0$ if, and
only if, $W(y'|x)>0$, then~$\Ceo$ and~$\Cal$ are unaltered when $y$ and $y'$ are combined into a
single output symbol distinct from all other output symbols. 
\end{proposition}
\begin{proof}
The set $\set{L}(\bfy)$ remains unchanged when any occurrence of $y$ in $\bfy$
is replaced with $y'$, or vice versa. 
\end{proof}
Using Proposition~\ref{prop:combine} we can also reduce the size of the output
alphabet to at most $2^{\card{\set{X}}}-1$ symbols while preserving~$\Ceo$ and
$\Cal$ (there are $2^{\card{\set{X}}}-1$ nonempty subsets of inputs).
In particular, every binary-input DMC can be reduced to an asymmetric binary
erasure channel (possibly with some transition probabilities equal to zero). 
And since the channel graph of the latter is acyclic, we can apply
Proposition~\ref{prop:factorize} to it. In this way we can determine~$\Ceo$
and~$\Cal$ for any binary-input channel.

%

\subsection{Relationships and Analogies}
There is a remarkable similarity between the way $\Coff$ relates to $C$ and the
way $\Cal$ relates to $\Ceo$. The following two propositions illustrate this.
The first is well-known~\cite{gallager1968information}.
\begin{proposition}
\label{prop:relationship_c_cutoff}
For every $\rho>0$, 
\begin{enumerate}
\item $\Coff \leq C$;
\item $\Coff>0$ $\iff$ $C>0$ $\iff$ there exist $x,x',y$ such that $W(y|x)\neq
W(y|x')$; \label{part:relationship_c_cutoff_2}
\item $\lim_{\rho\to0} \Coff = C$; \label{part:relationship_c_cutoff_3}
\item \label{part:relationship_c_cutoff_4} and $\lim_{\rho\to\infty} \Coff = -
\log \pi_0$, where
\begin{equation}
\label{eq:pi_0}
\pi_0=\min_{P} \max_{y\in\set{Y}} P\bigl(\set{X}(y)\bigr).
\end{equation}
\end{enumerate}
\end{proposition}
\begin{proof}
All assertions follow from~\eqref{eq:225_gallager}, the fact that $C=\max_P
I(P,W)$, and the properties of mutual information and Gallager's $E_0$ function
(see~\cite[Thm.~5.6.3]{gallager1968information} and Appendix~\ref{appendix:gallager}). 
\end{proof}

Proposition~\ref{prop:relationship_c_cutoff} remains (almost) true when $\Coff$ is
replaced with $\Cal$, and $C$ is replaced with $\Ceo$. 
\begin{proposition}
\label{prop:list_erasure_relation}
For every $\rho>0$,
\begin{enumerate}
\item \label{item:prop:list_erasure_relation_1} $\Cal \leq \Ceo$;
\item \label{item:prop:list_erasure_relation_2} $\Cal>0$ $\iff$ $\Ceo>0$ $\iff$ there exist $x,x',y$ such that
$W(y|x)>W(y|x')=0$; 
\item \label{item:prop:list_erasure_relation_3} $\lim_{\rho\to0} \Cal = \Ceo$; 
\item \label{item:prop:list_erasure_relation_4} and $\lim_{\rho\to\infty} \Cal =
-\log \pi_0$.
\end{enumerate}
\end{proposition}
\begin{proof}
Part~\ref{item:prop:list_erasure_relation_1} follows from Markov's inequality:
\begin{align}
\Pr\bigl(\card{\set{L}(\bfY)}\geq
2\bigr)&=\Pr\bigl(\card{\set{L}(\bfY)}^\rho-1\geq 2^\rho-1\bigr)\notag\\
&\leq \frac{\E{\card{\set{L}(\bfY)}^\rho}-1}{2^\rho-1},\label{eq:681}
\end{align}
and the right-hand side of~\eqref{eq:681} tends to zero if $\E{\card{\set{L}(\bfY)}^\rho}$
tends to one.

To prove Part~\ref{item:prop:list_erasure_relation_2}, assume that for every
$y\in\set{Y}$, $W(y|x)>0$
for some $x\in \set{X}$ implies $W(y|x')>0$ for all $x'\in\set{X}$. Then $\card{\set{L}(\bfy)}=e^{nR}$ for all 
sequences~$\bfy\in\set{Y}^n$ that can be produced by some (and hence all) messages.
Thus, $\Ceo=\Cal=0$. Conversely, if there exist~$x,x',y$ for 
which~$W(y|x)>W(y|x')=0$, then combine all outputs other than~$y$ into a single
output distinct from $y$, and use only the inputs $x$ and $x'$. This reduces the
channel to a Z-channel with crossover probability $1-W(y|x)$. For the
Z-channel we have by Proposition~\ref{prop:factorize} that $\Ceo= C$ and $\Cal =\Coff$,
where both $C$ and $\Coff$ are positive 
by Proposition~\ref{prop:relationship_c_cutoff} Part~\ref{part:relationship_c_cutoff_2}.

As to Part~\ref{item:prop:list_erasure_relation_3}, 
since $\Cal$ is clearly nonincreasing in $\rho$, the limit exists 
and is upper-bounded by~$\Ceo$ on account of Part~\ref{item:prop:list_erasure_relation_1}. 
On the other hand, it follows from the proof of~\cite[Thereom~1]{ahlswede1996erasure} that for
every rate $R<\Ceo$ the probability that $\card{\set{L}(\bfY)}$ exceeds one can be driven to zero
exponentially in the blocklength, i.e., we can find a sequence of 
blocklength-$n$ rate-$R$ encoders~$(f_n)_{n\geq 1}$
for which this probability is bounded by~$e^{-n\delta}$ for some (possibly very
small)~$\delta>0$.
For this sequence of encoders, 
\begin{align}
&\frac{1}{e^{nR}} \sum_{m=1}^{e^{nR}} \sum_{\bfy\in\set{Y}^n}
W^n\bigl(\bfy\big|f_n(m)\bigr)
\card{\set{L}(\bfy)}^\rho\notag\\
& \leq 1 + e^{n\rho R} \frac{1}{e^{nR}} \sum_{m=1}^{e^{nR}}
\sum_{\bfy:\card{\set{L}(\bfy)}\geq 2} W^n\bigl(\bfy\big|f_n(m)\bigr)\label{eq:1003_dfjsl}\\
&\leq 1+ e^{n\rho R} e^{-n\delta},\label{eq:634}
\end{align}
where~\eqref{eq:1003_dfjsl} follows by splitting the sum over
$\bfy\in \set{Y}^n$ into a sum over all $\bfy$ for 
which~$\card{\set{L}(\bfy)}=1$ and a sum over all other $\bfy$, and by using
$\card{\set{L}(\bfy)} \leq e^{nR}$ to bound the latter.  
Part~\ref{item:prop:list_erasure_relation_3} follows by noting that the right-hand side of~\eqref{eq:634} tends to one as
$n$ tends to infinity if $\rho< \delta/R$. 

As to Part~\ref{item:prop:list_erasure_relation_4}, since $\Cal \leq \Coff$, we
have $\lim_{\rho\to\infty} \Cal \leq - \log \pi_0$ by Proposition~\ref{prop:relationship_c_cutoff}
Part~\ref{part:relationship_c_cutoff_4}. On the other hand, 
it follows from~\eqref{eq:forney_cal_rho} by replacing the average over
$y\in\set{Y}$ with the maximum that $\Cal \geq  - \log \pi_0$ for all~$\rho>0$. 
\end{proof}
\begin{proposition}
\label{prop:list_erasure_relation_fb}
Propositions~\ref{prop:relationship_c_cutoff}
and~\ref{prop:list_erasure_relation} are true also for channels
with feedback. In particular, if $\Cal$, $\Calf$, $\Ceo$, or $\Ceof$ is
positive, then they all are.
\end{proposition}
\begin{proof}
In the case of Proposition~\ref{prop:relationship_c_cutoff} this follows
from the fact that feedback does not increase the Shannon capacity or the cutoff
rate (Corollary~\ref{cor:cutoff_feedback_no_increase}).
In the case of Proposition~\ref{prop:list_erasure_relation} the original 
proof goes through except for Part~\ref{item:prop:list_erasure_relation_3}. This
part, however, is contained in Corollary~\ref{cor:limit} ahead.
\end{proof}

The quantity $-\log \pi_0$ appearing in
Propositions~\ref{prop:relationship_c_cutoff}
and~\ref{prop:list_erasure_relation} has the following operational significance.
Shannon~\cite{shannon1956zero} proved that the zero-error capacity with feedback
$\Czerof$ can be expressed as 
\begin{equation}
\label{eq:shannon_zero}
\Czerof = \begin{cases} -\log \pi_0& \text{if $\Czero>0$,}\\0&
\text{otherwise.}\end{cases}
\end{equation}
He further conjectured that
\begin{equation}
\label{eq:shannon_conjecture}
-\log \pi_0 = \min_{V\ll W} C(V),
\end{equation}
where $C(V)$ denotes the Shannon capacity of the channel $V(y|x)$,
$x\in\set{X}$, $y\in\set{Y}$, and where, as above, $V\ll W$ means that $V(y|x) =0$ whenever $W(y|x)=0$. Ahlswede proved this
conjecture in~\cite{ahlswede1973channels}. Using the multiletter 
version of~\eqref{eq:const_comp_lb_cal}, Telatar~\cite{telatar1997zero} showed that
\begin{equation}
\label{eq:1056_telatar}
\lim_{\rho\to\infty} \Cal = \min_{V\ll W} C(V). 
\end{equation}
Combining~\eqref{eq:1056_telatar} with Part~\ref{item:prop:list_erasure_relation_4} 
of Proposition~\ref{prop:list_erasure_relation} furnishes an alternative
proof of~\eqref{eq:shannon_conjecture}.

\subsection{Sperner Capacity and $\epsilon$-Noise Channels}
There is an interesting relationship between the listsize capacity, the z.u.e.\
capacity, and the Sperner capacity of directed graphs~\cite{bunte2014zero}.
We say that a DMC is $\varepsilon$-noise if $\set{X}
\subseteq \set{Y}$ and 
\begin{equation}
W(x|x) \geq 1- \varepsilon,\quad \text{for all $x\in\set{X}$.}
\end{equation}
A natural way to associate a directed graph $G$ with an $\varepsilon$-noise
channel~$W$ is to take $\set{X}$ as the vertex set and to  introduce an edge
from $x$ to $y$ if $x\neq y$ and $W(y|x)>0$. 
It can be shown that~\cite{bunte2014zero,ahlswede1996erasure} 
\begin{equation}
\lim_{\varepsilon \to 0} \Ceo = \lim_{\varepsilon \to 0}
\Cal = \Sigma(G),
\end{equation}
where $\Sigma(G)$ denotes the Sperner capacity of $G$. (The limits are to be
understood in a uniform sense with respect to all $\varepsilon$-noise channels
with given graph $G$.) 

As a corollary to Theorem~\ref{thm:feedback_lb} ahead, we 
can show:
\begin{proposition}
\label{prop:low_noise_fb}
For any $\epsilon$-noise channel with $\set{X}=\set{Y}$ and $\Cal>0$,
\begin{equation}
\label{eq:low_noise_fb}
\lim_{\varepsilon\to0} \Calf = \log \card{\set{X}}.
\end{equation}
\end{proposition}
The proof of Proposition~\ref{prop:low_noise_fb} is postponed until 
Section~\ref{sec:feedback_lb}. 

\subsection{A Dual Source-Coding Problem}
A source coding analog to the listsize capacity has recently been 
studied in~\cite{bunte2014encoding}. There, the encoder uses~$nR$ bits to describe a sequence of
length~$n$ emitted by an IID source~$P_X$. Based on this
description, the decoder produces a list of sequences that is guaranteed
to contain the one emitted by the source. It is shown that the smallest rate~$R$ achievable in the
sense that the $\rho$-th moment of the length of this list tends to one as $n$
tends to infinity is given by the R\'enyi entropy of order $1/(1+\rho)$
\begin{equation}
H_{\frac{1}{1+\rho}}(X) = \frac{1}{\rho}\log \biggl(\sum_{x\in\set{X}}
P_X(x)^{\frac{1}{1+\rho}}\biggr)^{1+\rho}.
\end{equation}
It is also shown that if the source produces pairs $(X,Y)$ and the
$Y$-sequence is known as side-information at the encoder and
decoder, then the smallest achievable rate is given by a conditional version of
R\'enyi entropy
\begin{equation}
H_{\frac{1}{1+\rho}}(X|Y)= \frac{1}{\rho}\log \sum_{y\in\set{Y}}\biggl(\sum_{x\in\set{X}}
P_{X,Y}(x,y)^{\frac{1}{1+\rho}}\biggr)^{1+\rho}.
\end{equation}
This definition of conditional R\'enyi entropy was proposed by
Arimoto~\cite{arimoto1977information}, who showed that
\begin{equation}
\max_P \frac{E_0(\rho,P)}{\rho} = \max_{P}
H_{\frac{1}{1+\rho}}(X)-H_{\frac{1}{1+\rho}}(X|Y),
\end{equation}
where $(X,Y) \sim P(x)W(y|x)$. Thus, at least for channels
whose channel law factorizes in the Csisz\'ar-Narayan sense~\eqref{eq:factorize}, 
R\'enyi entropy plays a role in
channel and source coding with lists that is reminiscent of the role played by Shannon entropy in 
channel and source coding with the usual probability of error criteria. 

\section{The Converse Part of Theorem~\ref{thm:cutoff_no_increase}}
\label{sec:cutoff_fb}
In this section we prove the converse part of
Theorem~\ref{thm:cutoff_no_increase}, i.e., we prove
\begin{equation}
\label{eq:1340_converse}
\Cofff \leq \max_P \frac{E_0(\rho,P)}{\rho}.
\end{equation}
We need the following lemmas.
\begin{lemma}[\protect{\cite[Thm.~1]{arikan1996inequality}}]
\label{lem:arikan}
If the pair $(X,Y)\in\set{X}\times\set{Y}$ (where $\set{X}$ and $\set{Y}$ are
finite sets) has PMF $P_{X,Y}$, and if the function $G\colon
\set{X}\times\set{Y} \to \{1,\ldots, \card{\set{X}}\}$ is one-to-one as a
function of $x\in\set{X}$ for every $y\in\set{Y}$, then
\begin{equation}
\E{G(X,Y)^{\rho}} \geq \frac{1}{(1+\log \card{\set{X}})^\rho}  \sum_{y\in
\set{Y}} \Bigl( \sum_{x\in\set{X}}
P_{X,Y}(x,y)^{\frac{1}{1+\rho}}\Bigr)^{1+\rho}.
\end{equation}
\end{lemma}
\begin{lemma}[\protect{\cite[Thm.~5.6.5]{gallager1968information}}]
\label{lem:gallager}
A Necessary and sufficient condition for a PMF $P$ to minimize 
\begin{equation}
\sum_{y \in \set{Y}} \Bigl(\sum_{x\in \set{X}} P(x)
W(y|x)^{\frac{1}{1+\rho}} \Bigr)^{1+\rho}
\end{equation}
(and hence maximize $E_0(\rho,P)$)
is
\begin{equation}
\label{eq:gallager_kkt}
\sum_{y \in \set{Y}} W(y|x)^{\frac{1}{1+\rho}} \alpha_y(P)^\rho \geq
\sum_{y\in\set{Y}}
\alpha_y(P)^{1+\rho},
\end{equation}
for all $x\in\set{X}$, with equality if $P(x)>0$.
Here,
\begin{equation}
\alpha_y(P) = \sum_{x \in \set{X}} P(x) W(y|x)^{\frac{1}{1+\rho}}.
\end{equation}
\end{lemma}

Equipped with these lemmas, we can now prove~\eqref{eq:1340_converse}.
Fix a sequence of rate-$R$ blocklength-$n$ encoders as in
\eqref{eq:feedback_encoder}. For each $\bfy\in\set{Y}^n$ list the messages in decreasing order of their likelihood
(resolving ties arbitrarily)
\begin{equation}
\prod_{i=1}^{n} W\bigl(y_i\big|f_n^{(i)}(m,y^{i-1})\bigr),\quad 1\leq m\leq e^{nR}.
\end{equation}
Let $G(m,\bfy)$ denote the position of the $m$-th message in this list. 
Then $G(\cdot,\bfy)$ is one-to-one for every $\bfy\in\set{Y}^n$, and 
\begin{equation}
\label{eq:itw_cutoff_1}
G(m,\bfy) \leq \card{\set{L}(m,\bfy)},\quad 1\leq m \leq e^{nR}.
\end{equation}
(Equality holds in~\eqref{eq:itw_cutoff_1} if no message other than $m$ has the same likelihood as $m$.)
By Lemma~\ref{lem:arikan}, 
\begin{align}
&\frac{(1+nR)^\rho}{e^{nR}} \sum_{m=1}^{e^{nR}} \sum_{\bfy\in\set{Y}^n} G(m,\bfy)^\rho \prod_{i=1}^{n}
W\bigl(y_i\big|f_n^{(i)}(m,y^{i-1})\bigr)\notag\\
&\geq \sum_{\bfy\in\set{Y}^n} \Biggl( \sum_{m=1}^{e^{nR}}
\biggl(\frac{1}{e^{nR}}\prod_{i=1}^{n}
W\bigl(y_i\big|f^{(i)}_n(m,y^{i-1})\bigr)\biggr)^{\frac{1}{1+\rho}}
\Biggr)^{1+\rho}\notag \\
&= e^{n\rho R}\sum_{\bfy\in\set{Y}^n} \biggl(
\sum_{\bff\in\set{F}} \widetilde{P}(\bff) 
\widetilde{W}_n(\bfy|\bff)^{\frac{1}{1+\rho}}\biggr)^{1+\rho},\label{eq:itw_cutoff_2}
\end{align}
where $\widetilde{W}_n$ is the channel whose input alphabet $\set{F}$ is the set of all 
$n$-tuples $\bff=(f^{(1)},\ldots,f^{(n)})$ of functions of the form 
\begin{equation}
f^{(i)} \colon \set{Y}^{i-1} \to \set{X}, \quad i=1,\ldots,n;
\end{equation}
whose output alphabet is~$\set{Y}^n$; and whose
transition law is 
\begin{equation}
\label{eq:equivalent_channel}
\widetilde{W}_n(\bfy|\bff) = \prod_{i=1}^n
W\bigl(y_i\big|f^{(i)}(y^{i-1})\bigr),\quad \bfy\in\set{Y}^n,\,\bff\in\set{F},
\end{equation}
and where $\widetilde{P}$ is the PMF on $\set{F}$ induced by uniform 
messages and the encoding functions:
\begin{equation}
\label{eq:Q_tilde}
\widetilde{P}(\bff) = \frac{\bigcard{\bigl\{m:
\bigl(f^{(1)}_n(m),\ldots,f^{(n)}_n(m,\cdot)\bigr)=\bff\bigr\}}}{e^{nR}}.
\end{equation}
The proof is complete once we establish that
\begin{equation}
\label{eq:25_3_21_39}
\sum_{\bfy\in\set{Y}^n} \biggl( \sum_{\bff\in\set{F}} \widetilde{P}(\bff)
\widetilde{W}_n(\bfy|\bff)^{\frac{1}{1+\rho}} \biggr)^{1+\rho}
\geq e^{-n \max_{P} E_0(\rho,P)},
\end{equation}
because it will then follow using~\eqref{eq:itw_cutoff_1} and~\eqref{eq:itw_cutoff_2}
that the $\rho$-th moment of $\card{\set{L}(M,\bfY)}$ cannot tend to one
unless~\mbox{$R \leq \max_P E_0(\rho,P)/\rho$}.
To establish~\eqref{eq:25_3_21_39}, let $P^\star$ be a PMF on $\set{X}$ that minimizes 
\begin{equation}
\sum_{y\in\set{Y}} \biggl(\sum_{x\in\set{X}} P(x) W(y|x)^{\frac{1}{1+\rho}}\biggr)^{1+\rho}
\end{equation}
and hence achieves the maximum of $E_0(\rho,P)$. 
We use Lemma~\ref{lem:gallager} (applied to the channel $\widetilde{W}_n$) to show that the PMF $\widetilde{P}^\star$ on
$\set{F}$ given by
\begin{equation}
\widetilde{P}^\star(\bff) = \begin{cases} \prod_{i=1}^n P^\star(x_i)&
f^{(1)}\equiv x_1,\ldots,f^{(n)}\equiv x_n,\\0
&\text{otherwise,}\end{cases}
\end{equation}
minimizes the left-hand side of~\eqref{eq:25_3_21_39} over all PMFs on $\set{F}$. 
The notation $f^{(i)} \equiv x_i$
means that $f^{(i)}(y^{i-1}) = x_i$ for all $y^{i-1} \in \set{Y}^{i-1}$.
To verify that $\widetilde{P}^\star$ satisfies the conditions of Lemma~\ref{lem:gallager} for the channel
$\widetilde{W}_n$, observe that
\begin{align}
&\sum_{\bfy\in\set{Y}^n} \widetilde{W}_n(\bfy|\bff)^{\frac{1}{1+\rho}}
\biggl(\sum_{\bff'\in\set{F}}
\widetilde{P}^\star(\bff')
\widetilde{W}_n(\bfy|\bff')^{\frac{1}{1+\rho}}\biggr)^{\rho}\notag\\
&=\sum_{\bfy\in\set{Y}^n} \widetilde{W}_n(\bfy|\bff)^{\frac{1}{1+\rho}}
\biggl(\sum_{\bfx\in\set{X}^n} (P^\star)^n(\bfx)
W^n(\bfy|\bfx)^{\frac{1}{1+\rho}}\biggr)^\rho\notag\\ 
&=\sum_{y_1} W(y_1|f^{(1)})^{\frac{1}{1+\rho}}
\alpha_{y_1}(P^\star)^\rho\notag\\
&\qquad\times \cdots \times \sum_{y_n} W\bigl(y_n\big|f^{(n)}(y^{n-1})\bigr)^{\frac{1}{1+\rho}}
\alpha_{y_n}(P^\star)^\rho.\label{eq:nested_sums}
\end{align}
Applying \eqref{eq:gallager_kkt} (with $P$ replaced by $P^\star$) to the innermost of the nested sums  on the right-hand side
of~\eqref{eq:nested_sums} (the sum over~$y_n$), then to the second innermost (the sum over~$y_{n-1}$), and so on, we obtain
\begin{align}
&\sum_{\bfy\in\set{Y}^n} \widetilde{W}_n(\bfy|\bff)^{\frac{1}{1+\rho}}
\biggl(\sum_{\bff'\in\set{F}}
\widetilde{P}^\star(\bff')
\widetilde{W}_n(\bfy|\bff')^{\frac{1}{1+\rho}}\biggr)^{\rho}\notag\\
&\quad\geq \biggl(\sum_{y\in\set{Y}} \Bigl(
\sum_{x\in\set{X}} P^\star(x) W(y|x)^{\frac{1}{1+\rho}}
\Bigr)^{1+\rho}\biggr)^n\notag\\
&\quad= \sum_{\bfy\in\set{Y}^n} \biggl(
\sum_{\bfx\in\set{X}^n} (P^\star)^{n}(\bfx) W^n(\bfy|\bfx)^{\frac{1}{1+\rho}}
\biggr)^{1+\rho}\notag\\
&\quad= \sum_{\bfy\in\set{Y}^n} \biggl( \sum_{\bff'\in\set{F}}
\widetilde{P}^\star(\bff')
\widetilde{W}_n(\bfy|\bff')^{\frac{1}{1+\rho}}\biggr)^{1+\rho},
\end{align}
with equality if $f^{(1)}\equiv x_1,\ldots,f^{(n)}\equiv x_n$ and
\mbox{$P^\star(x_i)>0$} for all $i\in\{1,\ldots,n\}$, i.e., with equality if
$\widetilde{P}^\star(\bff)>0$. The PMF~$\widetilde{P}^\star$ thus satisfies the
conditions of Lemma~\ref{lem:gallager} (for the channel $\widetilde{W}_n$) for minimizing the left-hand side
of~\eqref{eq:25_3_21_39}, and the
value of this minimum is equal to the right-hand side of~\eqref{eq:25_3_21_39}.\qed

\section{The Direct Part of Theorem~\ref{thm:list_cutoff_achieve}}
\label{sec:cutoff_achieve}
\begin{figure}
\centering
\begin{tikzpicture}[scale=2, >=latex, thick]
\draw[-] (0,0) node[left] {$2$} -- node[above] {$\delta$} (2,0) node[right] {$2$};
\draw[-] (0,1) node[left] {$1$} -- node[above] {$1-\epsilon$} (2,1) node[right] {$1$};
\draw[-] (0,2) node[left] {$0$} -- node[above] {$1-\epsilon$} (2,2) node[right] {$0$};
\draw[-] (0,1) -- node[above, near start] {$\epsilon$} (2,2);
\draw[-] (0,2) -- node[above, near start] {$\epsilon$} (2,1);
\draw[-] (0,0) -- node[near start] {$\frac{1}{2}(1-\delta)$} (2,2);
\draw[-] (0,0) -- node[near start] {$\frac{1}{2}(1-\delta)$} (2,1);
\end{tikzpicture}
\caption{A channel with $0<\Calf< \max_P E_0(\rho,P)/\rho$.}
\label{fig:no_cutoff}
\end{figure}


Before presenting the proof of the direct part of
Theorem~\ref{thm:list_cutoff_achieve}, we comment on the necessity of the assumption~$\Czero>0$. 
Since the z.u.e.\ capacity with feedback is given by~\cite{nakiboglu2012errors,bunte2012zero}
\begin{equation}
\label{eq:nakiboglu}
\Ceof = \begin{cases} C& \text{if $\Ceo>0$,}\\0& \text{otherwise,}\end{cases}
\end{equation}
one might suspect that for equality in~\eqref{eq:list_cutoff_achieve} it
suffices that $\Cal$ be positive (and not necessarily $\Czero$). 
This, however, is not true: 

\begin{proposition}
A positive value of $\Cal$ does not guarantee equality
in~\eqref{eq:list_cutoff_achieve}.
\end{proposition}
\begin{proof}
A counterexample is the channel in Figure~\ref{fig:no_cutoff}. 
For this channel $\Czero=0$, 
$\Cal>0$, and $\max_P E_0(\rho,P)/\rho$ is at least close to~$\log 2$ for small
$\epsilon$. But even with feedback, if the received
sequence contains only the symbols~$0$ and~$1$, then the decoder cannot rule out
any of the messages and the list it produces is of size~$e^{nR}$. And regardless of
what is fed to the channel, the probability of observing only the symbols~$0$ and~$1$ 
at the output is at least $(1-\delta)^n$. Consequently, the~$\rho$-th moment of
the length of the list produced by the decoder is at least
\begin{equation}
e^{n\rho(R+\rho^{-1}\log(1-\delta))},
\end{equation}
and $\Calf$ must thus be bounded by $-\rho^{-1} \log(1-\delta)$, which is close
to zero for very small $\delta>0$ and hence smaller than $\max_P
E_0(\rho,P)/\rho$ if $\epsilon>0$ is sufficiently small. 
\end{proof}

To prove the direct part of Theorem~\ref{thm:list_cutoff_achieve}, we 
propose the following coding scheme. Let~$P^\star$ be a PMF on~$\set{X}$ that
achieves the maximum of $E_0(\rho,P)$. Select a sequence of types 
$(P_n)_{n\geq 1}$ with $P_n\to P^\star$ as $n\to\infty$, where each $P_n$ is a 
type in $\set{X}^n$.\footnote{This is
possible because the set of PMFs with rational components is dense in the set of
all PMFs.}
In the first phase, we send one of $e^{nR}$
messages using the length-$n$ type-$P_n$ codewords~$\bfx_1,\ldots,\bfx_{e^{nR}}$.
(We will generate the codebook at random later on). 
In the second phase, after the output sequence $\bfy\in\set{Y}^n$ has been observed through the feedback
link, we use a zero-error code (of rate at least $\log 2$) to describe the conditional type~$V$ of~$\bfy$ 
given the codeword.\footnote{To avoid uniqueness issues, we define the conditional type
$V(y|x)$ only for $x \in \set{X}$ with $P_n(x)>0$. Also, when the zero-error capacity is positive, then it
is at least $\log 2$. }
Since the number of conditional types is polynomial in $n$, this requires at most~$o(n)$ additional
channel uses. Let $\set{M}(\bfy,V)\subseteq\{1,\ldots,e^{nR}\}$ denote 
the set of all messages that are mapped to codewords given which $\bfy$ has conditional type $V$, i.e., 
\begin{equation}
\set{M}(\bfy,V) = \{1 \leq m \leq e^{nR}: \bfy \in T_V(\bfx_m)\}.
\end{equation}
At the end of the second phase both the encoder and the decoder know
$\set{M}(\bfy,V)$ and the decoder knows that the transmitted message is an
element of it. 
We fix some (small) $\alpha>0$ and partition~$\set{M}(\bfy,V)$ into~$e^{n\alpha}$
lists of lengths at most 
\begin{equation}
\bigl\lceil e^{-n\alpha}\card{\set{M}(\bfy,V)}\bigr\rceil.
\end{equation}
In the third phase, we send the index of the list containing the correct
message using a zero-error code (of rate at least $\log 2$). This requires at
most~$\lceil n \alpha/\log 2 \rceil$
additional channel uses. Note that the length of this list is determined by the
codeword and the first~$n$ channel outputs. We can upper-bound its $\rho$-th
moment by
\begin{equation}
\label{eq:287}
e^{-nR} \sum_{m=1}^{e^{nR}} \sum_{\bfy\in\set{Y}^n} W^n(\bfy|\bfx_m)
\bigl\lceil e^{-n\alpha}\card{\set{M}(\bfy,P_{\bfy|\bfx_m})}\bigr\rceil^\rho,
\end{equation}
where $P_{\bfy|\bfx_m}$ denotes the conditional type of $\bfy$ given $\bfx_m$.
Using the inequality
\begin{equation}
\lceil \xi \rceil^\rho <1+2^\rho \xi^\rho,\quad \xi\geq 0,
\end{equation}
we can upper-bound~\eqref{eq:287} by
\begin{equation}
\label{eq:297}
1+ 2^\rho e^{-n(R+\rho\alpha)} \sum_{m=1}^{e^{nR}} \sum_{\bfy\in\set{Y}^n} W^n(\bfy|\bfx_m)
\bigcard{\set{M}(\bfy,P_{\bfy|\bfx_m})}^\rho.
\end{equation}
Changing the order of summation, we can rewrite~\eqref{eq:297} as
\begin{multline}
\label{eq:297_2}
1+ 2^\rho e^{-n(R+\rho\alpha)}\\ \times \sum_{\bfy\in\set{Y}^n} \sum_{V}
\sum_{m\in \set{M}(\bfy,V)} W^n(\bfy|\bfx_m)
\bigcard{\set{M}(\bfy,V)}^\rho,
\end{multline}
where the middle sum extends over conditional types $V$. 
Using the identity
\begin{equation*}
W^n(\bfy|\bfx_m) = e^{-n(D(V||W|P_n)+H(V|P_n))},\quad \text{if $m \in
\set{M}(\bfy,V)$,}
\end{equation*}
and the fact that $\set{M}(\bfy,V)$ can be nonempty only if $\bfy$ has type~$P_nV$, 
we can rewrite~\eqref{eq:297_2} as
\begin{multline}
\label{eq:314}
1+2^\rho e^{-n(R+\rho\alpha)}\\ \times \sum\limits_{V} \sum\limits_{\bfy\in T_{P_nV}}
e^{-n(D(V||W|P_n)+H(V|P_n))} \card{\set{M}(\bfy,V)}^{1+\rho}.
\end{multline}
Next, we average the
upper bound~\eqref{eq:314} over all realizations of a random codebook
$\bfX_1,\ldots,\bfX_{e^{nR}}$ in which each codeword is drawn independently
and uniformly from $T_{P_n}$. This average is
\begin{multline}
\label{eq:316}
1+2^\rho e^{-n(R+\rho\alpha)}\\
\times \sum\limits_{V}\sum\limits_{\bfy\in T_{P_nV}} 
e^{-n(D(V||W|P_n)+H(V|P_n))}\bigE{\card{\set{M}(\bfy,V)}^{1+\rho}}.
\end{multline}
We now upper-bound the $(1+\rho)$-th moment of $\card{\set{M}(\bfy,V)}$. 
Under the given distribution of the codebook,
\begin{equation}
\card{\set{M}(\bfy,V)} = \sum_{m=1}^{e^{nR}} 1\bigl\{ \bfy \in T_V(\bfX_m)\bigr\}
\end{equation}
is a sum of IID Bernoulli random variables (RVs). To compute the probability 
of the event $\{\bfy\in T_V(\bfX_m)\}$ observe
that if $\bfx_m\in T_{P_n}$ and $\bfy \in T_{P_nV}$, then~$\bfy$ 
is in the $V$-shell of~$\bfx_m$ if, and only if,~$\bfx_m$
is in the~$\widetilde{V}$-shell of~$\bfy$, where 
\begin{equation}
\widetilde{V}(x|y) = \frac{V(y|x)P_n(x)}{(P_nV)(y)},\quad x\in\set{X},\, y\in
\supp(P_nV).
\end{equation}
Consequently, if $\bfy \in T_{P_nV}$, then 
\begin{align}
\Pr\bigl(\bfy\in T_V(\bfX_m)\bigr) &= \Pr\bigl(\bfX_m \in
T_{\widetilde{V}}(\bfy)\bigr)\notag\\
&= \frac{\card{T_{\widetilde{V}}(\bfy)}}{\card{T_{P_n}}}\label{eq:1534_gdkfsl}\\
&\leq e^{-n(H(P_n) - H(\widetilde{V}|P_nV) - \delta_n)}\label{eq:1683_boss}\\
&= e^{-n(I(P_n,V) - \delta_n)},\label{eq:991}
\end{align}
where~\eqref{eq:1534_gdkfsl} follows because $T_{\widetilde{V}}(\bfy) \subseteq
T_{P_n}$ when
$\bfy\in T_{P_nV}$, and because $\bfX_m$ is drawn uniformly at random
from~$T_{P_n}$;
where~\eqref{eq:1683_boss} follows because $\card{T_{\widetilde{V}}(\bfy)}
\leq e^{n H(\widetilde{V}|P_nV)}$ when $\bfy\in T_{P_nV}$, and because
$\card{T_{P_n}} \geq
e^{n(H(P_n)-\delta_n)}$; and where~\eqref{eq:991} follows by noting that
\begin{align}
H(P_n) - H(\widetilde{V}|P_nV) &= H(P_nV) - H(V|P_n)\notag\\
&=I(P_n,V).\label{eq:1544_dfsdf}
\end{align}
It is important to note that the $\delta_n$ appearing in~\eqref{eq:991} does not
depend on $V$. In fact, it can be taken as
\begin{equation}
\delta_n = \frac{\card{\set{X}} \log(n+1)}{n}.
\end{equation}
To bound the $(1+\rho)$-th moment of a binomial RV with exponential parameters,
we use Lemma~\ref{lem:binomial} (Appendix~\ref{sec:binomial}),
specifically~\eqref{eq:355}. 
This yields for every $\bfy \in T_{P_nV}$
\begin{multline}
\label{eq:350}
\bigE{\card{\set{M}(\bfy,V)}^{1+\rho}}\\ \leq \gamma e^{n(R-I(P_n,V)+\delta_n)} +
\gamma e^{n(1+\rho)(R-I(P_n,V)+\delta_n)}.
\end{multline}
Using~\eqref{eq:350}, the fact that $\card{T_{P_nV}} \leq e^{nH(P_nV)}$, 
and~\eqref{eq:1544_dfsdf}, we can upper-bound \eqref{eq:316} by
\begin{multline}
\label{eq:1009}
1+\gamma 2^\rho \sum_{V} e^{-n(\rho\alpha + D(V||W|P_n)-\delta_n)}\\
+ \gamma 2^\rho
\sum_{V} e^{-n( \rho\alpha -\rho R+ D(V||W|P_n) +\rho I(P_n,V) -(1+\rho)\delta_n)}.
\end{multline}
Since $D(V||W|P_n)$ is nonnegative, and since the number of conditional types
$V$ is polynomial in $n$, we can
upper-bound~\eqref{eq:1009} by
\begin{multline}
\label{eq:1016}
1+\gamma 2^\rho e^{-n(\rho\alpha  -\delta'_n)}\\ 
+\gamma 2^\rho e^{-n(\rho\alpha -\rho R+ \min_V D(V||W|P_n) +\rho
I(P_n,V) -\delta'_n)},
\end{multline}
where the minimum is over all channels $V(y|x)$, $x\in \set{X}$, $y\in \set{Y}$. 
The first exponential term tends to zero as $n$ tends to infinity if $\alpha>0$.
The second
exponential term tends to zero if
\begin{equation}
\label{eq:528}
R<  \alpha+ \varliminf_{n\to\infty}\min_{V}  I(P_n,V)+ \rho^{-1} D(V||W|P_n).
\end{equation}
The rate of the coding scheme approaches $R/(1+\alpha/\log2)$ as $n$ tends
to infinity. Letting $\alpha$ tend to zero, it thus follows 
that any rate below 
\begin{equation}
\label{eq:1613_min_bla}
\varliminf_{n\to\infty} \min_V I(P_n,V) + \rho^{-1} D(V||W|P_n)
\end{equation}
is achievable. And since by~\cite[Exercise 10.24]{csiszar2011information}
\begin{equation}
I(P_n,V) + \rho^{-1} D(V||W|P_n) \geq \frac{E_0(\rho,P_n)}{\rho},
\end{equation}
it follows from the continuity of $E_0(\rho,P)$ in $P$ that all rates 
below~$E_0(\rho,P^\star)/\rho$
are achievable.\qed

\section{A Lower Bound on the Listsize Capacity with Feedback}
\label{sec:feedback_lb}
The direct part of Theorem~\ref{thm:list_cutoff_achieve} is useless when $\Czero=0$. With this
case in mind, we propose
\begin{theorem}
\label{thm:feedback_lb}
If $\Cal>0$, then
\begin{equation}
\label{eq:feedback_lb}
\Calf \geq \frac{R^\star(\rho)}{1+\frac{\rho R^\star(\rho)}{\log \frac{1}{1-q^\star}}},
\end{equation}
where
\begin{equation}
\label{eq:1644_rstar}
R^\star(\rho) = \sup_{\xi > 0} \max_{P} \frac{E_0(\xi,P)}{\xi+\rho},
\end{equation}
and where $q^\star$ is the maximum of $W(\set{Y}_0|x_1)$ taken over all
$x_1\in\set{X}$ and over all the subsets $\set{Y}_0 \subset \set{Y}$ for which there
exists some $x_0\in\set{X}$ with $W(\set{Y}_0|x_0)=0$. If $\Czero>0$, i.e., if
the zero-error capacity is positive, then
$q^\star =1$, and we interpret the right-hand side of~\eqref{eq:feedback_lb} as~$R^\star(\rho)$. 
\end{theorem}
Note that the assumption $\Cal>0$ implies $q^\star>0$. Indeed, if $\Cal>0$, then, by 
Proposition~\ref{prop:list_erasure_relation}
Part~\ref{item:prop:list_erasure_relation_2}, we can find $x_0,x_1,y_0$ such that
$W(y_0|x_0)=0$ and $W(y_0|x_1)>0$. Taking $\set{Y}_0 = \{y_0\}$
thus shows that $q^\star \geq W(\set{Y}_0|x_1)>0$. 
Also note that, in view of Theorem~\ref{thm:list_cutoff_achieve} and
Proposition~\ref{prop:factorize}, the lower bound in~\eqref{eq:feedback_lb} is
interesting only when $\Czero=0$ and the channel law does not factorize in the
Csisz\'ar-Narayan sense~\eqref{eq:factorize}. 


Before presenting a proof of Theorem~\ref{thm:feedback_lb}, we use it to 
provide a proof of
Proposition~\ref{prop:low_noise_fb}, and we give another corollary to
Theorem~\ref{thm:feedback_lb}, Corollary~\ref{cor:limit}, 
which contains the earlier result~\eqref{eq:nakiboglu}. 

As to the proof of Proposition~\ref{prop:low_noise_fb}, in the notation of
Theorem~\ref{thm:feedback_lb} we have
$q^\star \geq 1-\epsilon$ if $W$ is $\epsilon$-noise and $\Cal>0$. 
Indeed, $\Cal>0$ implies that there exist $x_0$ and~$y_0$ such that $W(y_0|x_0)=0$
(Proposition~\ref{prop:list_erasure_relation} Part~\ref{item:prop:list_erasure_relation_2}), and the
$\epsilon$-noise property implies that $W(y_0|y_0)\geq 1-\epsilon$. 
Consequently, $\set{Y}_0 = \{y_0\}$, $x_1=y_0$ is a feasible choice in
the definition of $q^\star$. 
Moreover, if $P$ is the uniform PMF on~$\set{X}$, then
\begin{align}
&E_0(\xi,P)\notag\\
&= (1+\xi) \log \card{\set{X}} - \log \sum_{y\in\set{X}} \Bigl(
\sum_{x\in\set{X}}
W(y|x)^{\frac{1}{1+\xi}}\Bigr)^{1+\xi}\notag\\
&\geq (1+\xi)\log \card{\set{X}} - \log \sum_{y\in\set{X}} \Bigl(1+
\bigl(\card{\set{X}}-1\bigr)
\epsilon^{\frac{1}{1+\xi}}\Bigr)^{1+\xi}\notag\\
&=\xi \log \card{\set{X}} - (1+\xi)
\log\Bigl(1+\bigl(\card{\set{X}}-1\bigr)\epsilon^{\frac{1}{1+\xi}}\Bigr).\label{eq:1748_lownoise}
\end{align}
Now fix $\delta>0$ and choose $\xi>0$ large enough such that
$\xi/(\xi+\rho)> 1-\delta$. Then from~\eqref{eq:1644_rstar}
and~\eqref{eq:1748_lownoise} it follows that
\begin{equation}
R^\star(\rho) \geq (1-\delta) \log \card{\set{X}} - \frac{1+\xi}{\xi+\rho}
\log\Bigl(1+ \bigl(\card{\set{X}}-1\bigr) \epsilon^{\frac{1}{1+\xi}}\Bigr),
\end{equation}
and since the second term on the right-hand side tends to zero as $\epsilon \to 0$, it follows
from~\eqref{eq:feedback_lb} that
\begin{equation}
\liminf_{\epsilon \to 0} \Calf \geq (1-\delta) \log \card{\set{X}}.
\end{equation}
Letting $\delta \to 0$ thus proves~\eqref{eq:low_noise_fb}.\qed

\begin{corollary}
\label{cor:limit}
\begin{equation}
\label{eq:1145}
\lim_{\rho\to 0} \Calf = \Ceof = \begin{cases} C& \text{if
$\Ceo>0$,}\\0&\text{otherwise.}\end{cases}
\end{equation}
\end{corollary}

\begin{proof}
If $\Ceo=0$, then by Proposition~\ref{prop:list_erasure_relation_fb} also $\Ceof=0$ and
$\Calf=0$ for all $\rho>0$.
If $\Ceo>0$, then $\Cal>0$ and~\eqref{eq:feedback_lb} holds
for all $\rho>0$.
Moreover, from~\eqref{eq:1644_rstar} we have 
\begin{align}
\lim_{\rho\to 0} R^\star(\rho) &=\sup_{\rho>0} \sup_{\xi>0} \max_P
\frac{E_0(\xi,P)}{\xi+\rho}\notag\\
&= \max_P \sup_{\xi >0} 
\frac{E_0(\xi,P)}{\xi}\notag\\
&= \max_P I(P,W)\label{eq:1590_342}\\
&= C,\label{eq:1709_C}
\end{align}
where~\eqref{eq:1590_342} follows because $E_0(\xi,P)/\xi$ is nonincreasing in
$\xi>0$, $E_0(0,P)=0$, and $\partial E_0(\xi,P)/\partial \xi |_{\xi=0} =
I(P,W)$ (see~\cite[Thm.~5.6.3]{gallager1968information}). Consequently, by~\eqref{eq:feedback_lb}, $\lim_{\rho\to 0 } \Calf  \geq C$. And since $\Calf \leq
\Ceof \leq C$, it follows that $\lim_{\rho\to 0} \Calf = \Ceof = C$.  
\end{proof}

To prove Theorem~\ref{thm:feedback_lb}, we propose the following coding scheme. 
Select a positive integer $\ell$ and let $x_0,x_1,\set{Y}_0$ achieve~$q^\star$. 
In the first phase, we use a blocklength-$n$ rate-$R$ encoder paired with a
decoder that produces a list of the $\ell$ most likely messages given the
output of the channel (resolving ties arbitrarily).  
As shown in~\cite[Exercise 5.20]{gallager1968information}, for every PMF~$P$ on
$\set{X}$ we can find a sequence of
such encoders (indexed by the blocklength~$n$) such that the probability of the 
correct message not being on the list is
at most $e^{-n(E_0(\xi,P)-\xi R)}$ for every~$0 \leq \xi \leq \ell$. 

Thanks to the feedback, the transmitter knows which messages are on the
decoder's list, and in the second phase it tries to tell the receiver whether the
correct message is among them. To indicate that the correct message is on the
list, it sends~$n'$ times the symbol~$x_1$; otherwise it sends~$n'$ times the
symbol~$x_0$. 
Accordingly, if the receiver observes at least one symbol in $\set{Y}_0$ during the second
phase, it knows with certainty that the correct message is on the list
(because~$W(\set{Y}_0|x_0)=0$); otherwise it assumes that the correct message is
not on the list, it ignores the third phase, and it produces a final list
containing all $e^{nR}$ messages. 

If the first two phases are successful, i.e., if the list contains the correct
message and the receiver is aware of it, then the third phase 
is used to transmit the position of the correct message in the
list.
To this end, we construct~$\ell$ auxiliary codewords $\bfx_1,\ldots,\bfx_{\ell}$
of length~$k \ell$, where~$k$ is a fixed positive integer, as follows. The components $(i-1)k+1,\ldots, ik$ of $\bfx_i$
equal~$x_1$ and all its other components equal~$x_0$. 
The receiver can identify the correct auxiliary codeword, and thus produce the
correct message, if at least one symbol in~$\set{Y}_0$
is observed at the output during the third phase (because $W(\set{Y}_0|x_0)=0$ 
and the~$x_1$-patterns are disjoint). If no symbol in~$\set{Y}_0$ is observed
during the third phase, it produces the
list of size $\ell$ (which is guaranteed to contain the correct message). 
If the first or the second phase is unsuccessful, then it
does not matter what the transmitter does in the third phase. For concreteness,
it sends~$k \ell$ times the symbol~$x_0$. 

%

To analyze the performance of this coding scheme, define the events
\begin{align*}
E_1&=\{\textnormal{correct message not on the list after 1st phase}\},\\
E_2&=\{\textnormal{no symbol in $\set{Y}_0$ is observed in 2nd phase}\},\\
E_3&=\{\textnormal{no symbol in $\set{Y}_0$ is observed in 3rd phase}\}.
\end{align*}
Let $L$ be the length of the list produced by the receiver. The $\rho$-th moment
of $L$ is upper-bounded by
\begin{multline}
1+ \E{L^\rho|E_1} \Pr(E_1) + \E{L^\rho|\setcomp{E}_1\cap E_2}
\Pr(\setcomp{E}_1\cap E_2)\\
+ \E{L^\rho|\setcomp{E}_1\cap \setcomp{E}_2 \cap E_3}
\Pr(\setcomp{E}_1\cap\setcomp{E}_2\cap E_3)\label{eq:expected_list_size_decomp}.
\end{multline}
We upper-bound the right-hand side of~\eqref{eq:expected_list_size_decomp} term by term, beginning
with
\begin{align}
&\E{L^\rho|E_1} \Pr(E_1)\notag\\
&\quad\leq e^{n\rho R} e^{-n(E_0(\xi,P)-\xi R)}\notag\\
&\quad=e^{-n(\xi+\rho)\bigl(\frac{E_0(\xi,P)}{\xi+\rho}-R\bigr)},\quad 0\leq \xi \leq
\ell.\label{eq:lower_bound_first_cond_exp}
\end{align}
The right-hand side of \eqref{eq:lower_bound_first_cond_exp} approaches zero as $n$ tends to
infinity provided that $R < R^\star(\rho)$ and
$\ell$ is large enough so that we can pick a~$\xi$ in the interval $[0,\ell]$
and a PMF $P$ that achieve a value of
$E_0(\xi,P)/(\xi+\rho)$ close enough to $R^\star(\rho)$. The next term on the
right-hand side of~\eqref{eq:expected_list_size_decomp} can be upper-bounded as follows.
\begin{align}
E[L^\rho|\setcomp{E}_1\cap E_2] \Pr(\setcomp{E}_1\cap E_2) &\leq e^{n \rho R}
(1-q^\star)^{n'}\notag\\
&= e^{n\bigl(\rho R -\frac{n'}{n} \log\frac{1}{1-q^\star}\bigr)}.\label{eq:lower_bound_second_cond_exp}
\end{align}
The right-hand side of \eqref{eq:lower_bound_second_cond_exp} approaches zero as
$n$ tends to infinity if we choose
\begin{equation}
\label{eq:ratio_of_blocklengths}
n' = n(1+\delta)\frac{\rho R}{ \log \frac{1}{1-q^\star}}
\end{equation}
for an arbitrarily small $\delta>0$. (If $\Czero>0$, and hence $q^\star=1$, then we
may take $n'=1$.) 
Finally, 
\begin{equation}
E[L^\rho|\setcomp{E}_1 \cap \setcomp{E}_2 \cap
E_3]\Pr(\setcomp{E}_1\cap\setcomp{E}_2\cap E_3) \leq \ell (1-q^\star)^k,
\end{equation}
and the right-hand side can be made arbitrarily small by choosing $k$ sufficiently large.
(If $\Czero>0$, we may take $k=1$.)
The rate of the coding scheme is
\begin{equation}
\frac{R}{1+ \frac{n'}{n}+ \frac{k\ell}{n}}.
\end{equation}
Choosing first~$\ell$ sufficiently large, then~$R$ close to~$R^\star(\rho)$, then~$n'$ as
in~\eqref{eq:ratio_of_blocklengths} with~$\delta$ sufficiently small, then~$k$
sufficiently large, and finally~$n$ sufficiently large shows that that all rates
strictly less than the right-hand side of \eqref{eq:feedback_lb} are
achievable.\qed

\section{A Proof of Proposition~\ref{prop:forney_cal_rho} and 
the Asymptotic Tightness of~\eqref{eq:forney_cal_rho_multi}}
\label{sec:forney}
In this section we derive the lower bound~\eqref{eq:forney_cal_rho} and show 
that its $n$-letter version~\eqref{eq:forney_cal_rho_multi} becomes tight as
$n$ tends to infinity even when $P$ is restricted to be uniform over its support. 

We begin with a proof of~\eqref{eq:forney_cal_rho}.
Given a \mbox{blocklength-$n$} \mbox{rate-$R$} codebook $\bfx_1,\ldots,\bfx_{e^{nR}}$, we can
write the $\rho$-th moment of~$\card{\set{L}(\bfY)}$ as
\begin{equation}
\label{eq:1244}
\frac{1}{e^{nR}} \sum_{m=1}^{e^{nR}} \sum_{\bfy\in\set{Y}^n} W^n(\bfy|\bfx_m)\biggl(1+
\sum_{m'\neq m} Z_{m'}(\bfy)\biggr)^\rho,
\end{equation}
where we define
\begin{equation}
Z_m(\bfy) = 1\bigl\{W^n(\bfy|\bfx_m)>0\bigr\},\quad 1\leq m \leq e^{nR}.
\end{equation}
If the codebook is generated at random by drawing each component of each codeword 
independently according to a PMF~$P$ on~$\set{X}$, then the expectation of~\eqref{eq:1244} (over the codebook) is
\begin{equation}
\label{eq:1251}
\sum_{\bfy\in\set{Y}^n} (PW)^n(\bfy) \BiggE{\biggl(1+\sum_{m=2}^{e^{nR}} Z_m(\bfy)
\biggr)^\rho},
\end{equation}
where for every $\bfy\in\set{Y}^n$ the RVs $Z_1(\bfy),\ldots,Z_{e^{nR}}(\bfy)$
are IID Bernoulli. Note that
\begin{align}
\Pr\bigl(Z_m(\bfy) =1\bigr)&=\prod_{i=1}^n P\bigl(\set{X}(y_i)\bigr)\notag\\
&=\prod_{y\in\set{Y}} P\bigl(\set{X}(y)\bigr)^{nP_\bfy(y)}\notag\\
&=e^{n \sum_{y\in\set{Y}} P_\bfy(y) \log P(\set{X}(y))}\notag\\
&=e^{-n F(P_\bfy)},
\end{align}
where $P_\bfy$ is the type of $\bfy$, and where we define
\begin{equation}
F(Q) = -\sum_{y\in\set{Y}} Q(y) \log P\bigl(\set{X}(y)\bigr).
\end{equation}
To prove~\eqref{eq:forney_cal_rho} it suffices to show
that~\eqref{eq:1251} tends to one as $n$ tends to infinity whenever
\begin{equation}
\label{eq:2030_blabla}
R< - \rho^{-1} \log \sum_{y\in\set{Y}} (PW)(y) P\bigl(\set{X}(y)\bigr)^\rho.
\end{equation}
We first show that~\eqref{eq:2030_blabla} is equivalent to  
\begin{equation}
\label{eq:1274}
R<\min_Q  F(Q) + \rho^{-1} D(Q||PW).
\end{equation}
where the minimum is over all PMFs $Q$ on $\set{Y}$. 
Indeed, observe that
\begin{align}
&F(Q) + \rho^{-1} D(Q||PW)\notag\\
&\quad= - \rho^{-1} \sum_{y\in\set{Y}} Q(y) \log \frac{
(PW)(y)P(\set{X}(y))^\rho}{Q(y)}\notag\\
&\quad\geq -\rho^{-1} \log \sum_{y\in\set{Y}} (PW)(y)
P(\set{X}(y))^\rho,\label{eq:2148_jensen}
\end{align}
where~\eqref{eq:2148_jensen} follows from Jensen's Inequality. 
The choice
\begin{equation}
Q(y) = \frac{(PW)(y) P(\set{X}(y))^\rho}{\sum_{y'\in\set{Y}}
(PW)(y')P(\set{X}(y'))^\rho},\quad y\in\set{Y},
\end{equation}
achieves equality in~\eqref{eq:2148_jensen}.

Using Lemma~\ref{lem:binomial} (Appendix~\ref{sec:binomial}),
specifically~\eqref{eq:1547}, 
we can upper-bound~\eqref{eq:1251} by
\begin{multline}
\label{eq:1796_3420}
1+\gamma \sum_{\bfy\in\set{Y}^n}(PW)^n(\bfy) \Bigl(e^{n(R- F(P_\bfy))}1\{R \leq
F(P_{\bfy})\}\\
+ e^{n\rho(R- F(P_\bfy))}1\{R > F(P_{\bfy})\}\Bigr).
\end{multline}
Since $(PW)^n(T_Q) \leq e^{-nD(Q||PW)}$, we can upper-bound the sum
in~\eqref{eq:1796_3420} by 
\begin{multline}
\label{eq:1803_3908}
\sum_{Q:R\leq F(Q)}  e^{n(R- F(Q) - D(Q||PW))}\\
+\sum_{Q:R>F(Q)}
 e^{n\rho (R- F(Q) - \rho^{-1} D(Q||PW))},
\end{multline}
where $Q$ runs over all types in $\set{Y}^n$. 
Next, we show that if the rate $R$ satisfies~\eqref{eq:1274}, then~\eqref{eq:1803_3908} tends to
zero as $n$ tends to infinity.
Assume therefore that~\eqref{eq:1274} holds and define the positive number
\begin{equation}
\delta= \min_Q  F(Q) + \rho^{-1} D(Q||PW)-R. 
\end{equation}
The second sum in~\eqref{eq:1803_3908} tends to zero as $n$ tends to infinity
because the summand is upper-bounded by $e^{-n\rho \delta}$ and the number of
different types is polynomial in $n$. 
To show that the first sum in~\eqref{eq:1803_3908} tends to zero, we consider separately the cases
$\rho \geq 1$ and $\rho<1$. In the former case, the summand is upper-bounded by
$e^{-n \delta}$ because $D(Q||PW) \geq \rho^{-1} D(Q||PW)$. In the latter
case, the summand is upper-bounded by $e^{-n \rho \delta}$ because 
$R-F(Q) \leq \rho(R-F(Q))$ when $R\leq F(Q)$. 
We conclude that~\eqref{eq:1803_3908} tends to zero as $n$ tends to infinity 
for all rates $R$ satisfying~\eqref{eq:1274}. 
In view of the equivalence of~\eqref{eq:2030_blabla} and~\eqref{eq:1274}, this
completes the proof of~\eqref{eq:forney_cal_rho}.

To prove that~\eqref{eq:forney_cal_rho_multi} is asymptotically tight even when
$P$ is restricted to be uniform over its support, we define
\begin{equation}
J_n(\rho,P) = -\frac{1}{n\rho}\log \sum_{\bfy\in\set{Y}} (PW^n)(\bfy)
P\bigl(\set{X}^n(\bfy)\bigr)^\rho.
\end{equation}
Since~\eqref{eq:forney_cal_rho_multi} holds for every $n$, and since restricting
the feasible set cannot help, 
\begin{equation}
\label{eq:1476}
\Cal \geq \varlimsup_{n\to\infty} \max_{P\in\set{U}_n} J_n(\rho,P),
\end{equation}
where $\set{U}_n$ denotes the set of PMFs on $\set{X}^n$ that are uniform over
their support. 
It remains to show that 
\begin{equation}
\label{eq:2069_show}
\Cal \leq \varliminf_{n\to\infty} \max_{P\in\set{U}_n} J_n(\rho,P).
\end{equation}
To this end, fix a sequence of rate-$R$ blocklength-$n$ encoders~$(f_n)_{n\geq1}$ with 
\begin{equation}
\label{eq:1299}
e^{-nR} \sum_{m=1}^{e^{nR}} \sum_{\bfy\in\set{Y}^n} W^n\bigl(\bfy|f_n(m)\bigr)
\card{\set{L}(\bfy)}^\rho \leq  1+\epsilon_n,
\end{equation}
where $\epsilon_n\to0$ as $n\to\infty$. 
We first argue that the number of codewords to which only one message is mapped
by $f_n$ is at least $e^{n(R-\delta_n)}$. 
Indeed, if $m\neq m'$
and $f_n(m)=f_n(m')$, then $\card{\set{L}(\bfy)}\geq 2$ whenever $W^n(\bfy|f_n(m))>0$
(because then also $W^n(\bfy|f_n(m'))>0$), and hence
\begin{equation}
\label{eq:1306}
\sum_{\bfy\in\set{Y}^n} W^n\bigl(\bfy|f_n(m)\bigr) \card{\set{L}(\bfy)}^\rho \geq 2^\rho.
\end{equation}
If we define 
\begin{equation*}
\set{M}_n=\bigl\{1\leq m \leq e^{nR}: \text{$f_n(m')\neq f_n(m)$ for all $m' \neq
m$}\bigr\},
\end{equation*}
then it follows from~\eqref{eq:1299} and~\eqref{eq:1306} that
\begin{align}
e^{-nR} \card{\setcomp{\set{M}}_n} 2^\rho &\leq e^{-nR}\sum_{m=1}^{e^{nR}}
\sum_{\bfy\in\set{Y}^n} W^n\bigl(\bfy|f_n(m)\bigr)
\card{\set{L}(\bfy)}^\rho\notag\\
&\leq 1+\epsilon_n,  \label{eq:2067_rearrange}
\end{align}
where $\setcomp{\set{M}}_n$ denotes the set complement of $\set{M}_n$ 
in~$\{1,\ldots, e^{nR}\}$. 
Rearranging~\eqref{eq:2067_rearrange} gives
\begin{equation}
\label{eq:1964_engl}
\card{\setcomp{\set{M}}_n} \leq e^{nR} 2^{-\rho}(1+\epsilon_n).
\end{equation}
Since $\epsilon_n \to 0$ as $n\to\infty$, there exists $n_0$ such that $2^{-\rho}(1+\epsilon_n)<1$ for all $n\geq n_0$. 
Henceforth assume that~$n\geq n_0$. Since
$\card{\set{M}_n}+\card{\comp{\set{M}}_n}=e^{nR}$, it follows
from~\eqref{eq:1964_engl} that
\begin{align}
\card{\set{M}_n} &\geq e^{nR}\bigl(1-
2^{-\rho}(1+\epsilon_n)\bigr)\notag\\
&=e^{n(R-\delta_n)}.\label{eq:1971_fsds}
\end{align}
Since $1+\epsilon_n < 2^{\rho}$, restricting the
message set to $\set{M}_n$ can only decrease the $\rho$-th moment of
the length of the list, so
\begin{equation}
\label{eq:1332}
\frac{1}{\card{\set{M}_n}} \sum_{m\in\set{M}_n}\sum_{\bfy\in\set{Y}^n}
W^n\bigl(\bfy|f_n(m)\bigr) \card{\tilde{\set{L}}(\bfy)}^\rho \leq 1+\epsilon_n,
\end{equation}
where 
\begin{equation}
\tilde{\set{L}}(\bfy) =\bigl\{m\in\set{M}_n: W^n(\bfy|f_n(m))>0\bigr\}.
\end{equation}
Let $P_n$ be the uniform PMF on the set $\{f_n(m):m\in\set{M}_n\}$. Then $P_n
\in \set{U}_n$ and 
\begin{align}
P_n\bigl(\set{X}^n(\bfy)\bigr)&=
\frac{\card{\tilde{\set{L}}(\bfy)}}{\card{\set{M}_n}}\notag\\
&\leq e^{-n(R-\delta_n)} \card{\tilde{\set{L}}(\bfy)},\label{eq:1990_dfse}
\end{align}
where~\eqref{eq:1990_dfse} follows from~\eqref{eq:1971_fsds}.
Rearranging~\eqref{eq:1990_dfse} gives
\begin{equation}
\label{eq:1487}
\card{\tilde{\set{L}}(\bfy)} \geq e^{n(R-\delta_n)}
P_n\bigl(\set{X}^n(\bfy)\bigr).
\end{equation}
Combining~\eqref{eq:1487} and~\eqref{eq:1332}, and taking logarithms, we obtain
\begin{equation}
\label{eq:1293}
\log(1+\epsilon_n) \geq  n\rho (R-\delta_n) - n\rho J_n(\rho,P_n). 
\end{equation}
Dividing by $n\rho$ and letting $n\to\infty$ shows that
\begin{equation}
\label{eq:1972_bfd}
R \leq \varliminf_{n\to\infty} J_n(\rho,P_n). 
\end{equation}
The right-hand side of~\eqref{eq:1972_bfd} is upper-bounded by the right-hand side
of~\eqref{eq:2069_show} because $P_n \in \set{U}_n$.\qed

\appendices
\section{Exponential Upper Bounds on the $\rho$-th Moment of Binomial RVs}
\label{sec:binomial}
\begin{lemma}
\label{lem:binomial}
Let $X_1,\ldots,X_{e^{n\alpha}}$ be IID Bernoulli RVs with success probability 
\begin{equation}
\label{eq:success_p_bound}
p_n=\Pr(X_i=1) =1-\Pr(X_i=0) \leq e^{-n\beta}, 
\end{equation}
where $n\in\naturals$, $\alpha>0$ and $\beta\geq 0$. Let $\rho>0$. Then 
\begin{equation}
\label{eq:1547}
\BiggE{\biggl(1+\sum_{i=1}^{e^{n\alpha}} X_i\biggr)^\rho}\leq \begin{cases} 1 +
\gamma e^{n(\alpha-\beta)}& \text{if $\beta \geq \alpha$,}\\ \gamma
e^{n\rho(\alpha-\beta)}& \text{if $\beta<\alpha$,}\end{cases}
\end{equation}
and
\begin{equation}
\BiggE{ \biggl(\sum_{i=1}^{e^{n\alpha}} X_i\biggr)^{\rho}} \leq \begin{cases}
\gamma e^{n(\alpha-\beta)}& \text{if $\beta\geq \alpha$,} \label{eq:355}
\\\gamma e^{n\rho(\alpha-\beta)}& \text{if $\beta < \alpha$,}  
\end{cases}
\end{equation}
where
\begin{equation}
\label{eq:2125_gamma}
\gamma=\max\bigl\{e^{e^{\rho}-1}, (\lceil \rho \rceil !)^2 \lceil \rho \rceil\bigr\}.
\end{equation}
\end{lemma}
\begin{proof}
We use the inequalities
\begin{equation}
\label{eq:1778}
\xi< 1+\xi \leq e^\xi,\quad \xi\in\reals, 
\end{equation}
and
\begin{equation}
\label{eq:exp_convex_ineq}
e^{\eta \xi} \leq 1+ \xi (e^{\eta}-1),\quad 0\leq \xi \leq 1, \; \eta>0.
\end{equation}
(The inequality~\eqref{eq:exp_convex_ineq} is a consequence of the convexity of
the function $\xi \mapsto e^{\eta \xi}$.)

We begin with a proof of~\eqref{eq:1547}. Consider first the case $\beta \geq \alpha$ and observe that
\begin{align}
\BiggE{\biggl(1+\sum_{i=1}^{e^{n\alpha}} X_i\biggr)^\rho} &\leq
\BiggE{\exp\biggl(\rho \sum_{i=1}^{e^{n\alpha}} X_i\biggr)}\label{eq:2233_ineq}\\
&= \bigE{e^{\rho X_1}}^{e^{n\alpha}}\label{eq:2073_34980}\\
&=\bigl(1+p_n (e^\rho-1)\bigr)^{e^{n\alpha}}\notag\\
&\leq \exp\bigl( p_n e^{n\alpha} (e^\rho-1)\bigr)\label{eq:2236_ineq}\\
&\leq \exp\bigl( e^{n(\alpha-\beta)} (e^\rho-1)\bigr)\label{eq:2338_balling}\\
&\leq 1+ e^{n(\alpha-\beta)} (e^{e^\rho-1}-1)\label{eq:2238_ineq}\\
&\leq 1+ \gamma e^{n(\alpha-\beta)},\label{eq:2156_ineq}
\end{align}
where~\eqref{eq:2233_ineq} and~\eqref{eq:2236_ineq} follow from~\eqref{eq:1778};
where~\eqref{eq:2073_34980} follows because the $X_i$'s are IID;
where~\eqref{eq:2338_balling} follows from~\eqref{eq:success_p_bound}; 
where~\eqref{eq:2238_ineq}
follows from~\eqref{eq:exp_convex_ineq} with $\eta=e^\rho-1$ and
$\xi=e^{n(\alpha-\beta)}$; and where~\eqref{eq:2156_ineq} follows
from~\eqref{eq:2125_gamma}.

Now consider the case $\beta <\alpha $ and observe that
\begin{align}
&\BiggE{\biggl(1+\sum_{i=1}^{e^{n\alpha}} X_i\biggr)^\rho}\notag\\
&= e^{n\rho(\alpha-\beta)} \BiggE{\biggl(e^{-n(\alpha-\beta)}+e^{-n(\alpha-\beta)}
\sum_{i=1}^{e^{n\alpha}} X_i\biggr)^\rho}\notag\\
&\leq e^{n\rho(\alpha-\beta)} \BiggE{\biggl(1+e^{-n(\alpha-\beta)}
\sum_{i=1}^{e^{n\alpha}} X_i\biggr)^\rho}\label{eq:2171_ineq}\\
&\leq e^{n\rho(\alpha-\beta)} \BiggE{\exp\biggl(\rho e^{-n(\alpha-\beta)}
\sum_{i=1}^{e^{n\alpha}} X_i\biggr)}\label{eq:2255_ineq_2}\\
&=e^{n\rho(\alpha-\beta)}\bigE{ \exp(\rho e^{-n(\alpha-\beta)}
X_1)}^{e^{n\alpha}}\notag\\
&=e^{n\rho(\alpha-\beta)}\Bigl(1+p_n\bigl(\exp(\rho
e^{-n(\alpha-\beta)})-1\bigr)\Bigr)^{e^{n\alpha}}\notag\\
&\leq e^{n\rho(\alpha-\beta)} \exp \Bigl(p_n e^{n\alpha} \bigl(\exp(\rho
e^{-n(\alpha-\beta)})-1\bigr)\Bigr)\label{eq:2260_ineq}\\
&\leq e^{n\rho(\alpha-\beta)} \exp \Bigl(e^{n(\alpha-\beta)} \bigl(\exp(\rho
e^{-n(\alpha-\beta)})-1\bigr)\Bigr)\label{eq:1578}\\
&\leq e^{n\rho(\alpha-\beta)} e^{e^{\rho}-1}\label{eq:2263_ineq}\\
&\leq \gamma e^{n\rho(\alpha-\beta)},\label{eq:2182_ineq}
\end{align}
where~\eqref{eq:2171_ineq} follows because $e^{-n(\alpha-\beta)} \leq 1$; 
where~\eqref{eq:2255_ineq_2} and~\eqref{eq:2260_ineq} follow from~\eqref{eq:1778}; 
where~\eqref{eq:1578} follows from~\eqref{eq:success_p_bound};
where~\eqref{eq:2263_ineq} follows from~\eqref{eq:exp_convex_ineq} with
$\eta=\rho$ and $\xi= e^{-n(\alpha-\beta)}$; and where~\eqref{eq:2182_ineq}
follows from~\eqref{eq:2125_gamma}.

We now prove~\eqref{eq:355}. The case $\beta<\alpha$ is implied
by~\eqref{eq:1547}, and we only need to treat the case~$\beta\geq \alpha$. 
We first show that~\eqref{eq:355} holds when $\rho$ is an arbitrary positive
integer, which we denote by $k$. 
For any such~$k$, 
\begin{equation}
\label{eq:1431}
\biggl(\sum_{i=1}^{e^{n\alpha}} X_i\biggr)^{k} = \sum
\binom{k}{k_1,\ldots,k_{e^{n\alpha}}} 
\prod_{i=1}^{e^{n\alpha}} X_i^{k_i},
\end{equation}
where the sum on the right-hand side extends over all possible choices of nonnegative integers
$k_1,\ldots,k_{e^{n\alpha}}$ that sum up to~$k$. Taking the expectation on both
sides of~\eqref{eq:1431} yields
\begin{equation}
\label{eq:1440}
\BiggE{\biggl(\sum_{i=1}^{e^{n\alpha}} X_i\biggr)^{k}} = \sum
\binom{k}{k_1,\ldots,k_{e^{n\alpha}}} 
\prod_{i=1}^{e^{n\alpha}} \bigE{X_i^{k_i}},
\end{equation}
where we used the independence of the $X_i$'s.
Since the $X_i$'s are 0--1 valued, we
have $X_i^{k_i} = X_i$ if $k_i\geq 1$, and $X_i^{k_i} = 1$ if $k_i=0$.
Since the $X_i$'s have identical distributions, we thus have
\begin{equation}
\label{eq:1449}
\prod_{i=1}^{e^{n\alpha}} \bigE{X_i^{k_i}}= \E{X_1}^{\card{\{i: k_i\geq 1\}}}.
\end{equation}
Using the trivial upper bound
\begin{equation}
\binom{k}{k_1,\ldots,k_{e^{n\alpha}}} \leq k!,
\end{equation}
and substituting~\eqref{eq:1449} into~\eqref{eq:1440}, we obtain
\begin{equation}
\label{eq:409}
\BiggE{\biggl(\sum_{i=1}^{e^{n\alpha}} X_i\biggr)^{k}} \leq  k!
\sum_{k_1+\ldots+k_{e^{n\alpha}}=k}\E{X_1}^{\card{\{i: k_i\geq 1\}}}.
\end{equation}
For any choice of nonnegative integers $k_1,\ldots,k_{e^{n\alpha}}$ that sum up to
$k$, the number of indices $i$ for which $k_i\geq 1$ must be between~$1$ and~$k$, so we
may rewrite~\eqref{eq:409} as
\begin{equation}
\label{eq:2011_5839}
\BiggE{\biggl(\sum_{i=1}^{e^{n\alpha}} X_i\biggr)^{k}}\leq k! \sum_{\ell=1}^k \binom{e^{n\alpha}}{\ell} \binom{k-1}{\ell-1}\E{X_1}^\ell,
\end{equation}
where the first binomial coefficient accounts for the number of ways we can choose
exactly~$\ell$ of the~$e^{n\alpha}$ integers $k_1,\ldots,k_{e^{n\alpha}}$ to be positive, and where the second binomial
coefficient accounts for the number of ways we can choose the values of $\ell$
positive integers that sum up to~$k$. Upper-bounding~$\binom{k-1}{\ell-1}$ by~$k!$
and upper-bounding~$\binom{e^{n\alpha}}{\ell}$ by~$e^{n\ell
\alpha}$,~\eqref{eq:2011_5839} becomes
\begin{align}
\BiggE{\biggl(\sum_{i=1}^{e^{n\alpha}} X_i\biggr)^{k}} &\leq (k!)^2 \sum_{\ell=1}^k
e^{n\ell \alpha}\E{X_1}^\ell\notag\\
&\leq (k!)^2 \sum_{\ell=1}^k e^{n\ell (\alpha-\beta)}\notag\\
&\leq e^{n(\alpha-\beta)} (k!)^2 k.\label{eq:2255_ineq}
\end{align}
This proves~\eqref{eq:355} for $\beta\geq \alpha$ and all nonnegative integer values of
$\rho$. If $\beta \geq \alpha$ but $\rho$ is not an integer, then 
\begin{align}
\BiggE{\biggl(\sum_{i=1}^{e^{n\alpha}} X_i\biggr)^{\rho}} &\leq
\BiggE{\biggl(\sum_{i=1}^{e^{n\alpha}} X_i\biggr)^{\lceil \rho\rceil}}\notag\\
& \leq e^{n(\alpha-\beta)} (\lceil \rho\rceil!)^2 \lceil \rho
\rceil,\label{eq:2262_ineq}\\
& \leq \gamma e^{n(\alpha-\beta)},\label{eq:2264_ineq}
\end{align}
where~\eqref{eq:2262_ineq} follows from~\eqref{eq:2255_ineq}, and
where~\eqref{eq:2264_ineq} follows from~\eqref{eq:2125_gamma}. This completes
the proof of~\eqref{eq:355}.
\end{proof}

\section{A Proof of the Direct Part of~\eqref{eq:225_gallager}}
\label{appendix:cutoff}
Here we prove the achievability part of~\eqref{eq:225_gallager}, i.e., we
prove that for all $\rho>0$, 
\begin{equation}
\label{eq:gallager_direct}
\Coff \geq \max_P \frac{E_0(\rho,P)}{\rho}.
\end{equation}
Fix $\rho>0$ and a PMF $P$ on $\set{X}$. Generate a random 
\mbox{blocklength-$n$} rate-$R$ codebook $\bfX_1,\ldots,\bfX_{e^{nR}}$ by drawing each
component of each codeword independently according to~$P$. 
It suffices to show that the expectation of
\begin{equation}
\label{eq:1879}
\frac{1}{e^{nR}} \sum_{1\leq m \leq e^{nR}} \sum_{\bfy\in\set{Y}^n}
W^n(\bfy|\bfX_m) \card{\set{L}(m,\bfy)}^\rho
\end{equation}
(with respect to the distribution of the codebook) 
tends to one as $n$ tends to infinity when $R<E_0(\rho,P)/\rho$. 
This expectation can be expressed as 
\begin{align}
\label{eq:1894}
\sum_{\bfy\in\set{Y}^n}\sum_{\bfx_1\in\set{X}^n} W^n(\bfy|\bfx_1) P^n(\bfx_1) \BiggE{\biggl(1+\sum_{m=2}^{e^{nR}}
B_m(\bfy,\bfx_1)\biggr)^\rho},
\end{align}
where we define the RVs
\begin{equation}
B_m(\bfy,\bfx) = 1\bigl\{ W^n(\bfy|\bfX_m) \geq W^n(\bfy|\bfx)\bigr\}.
\end{equation}
Note that the distribution of $B_m(\bfy,\bfx)$ depends on $\bfx$ and $\bfy$ only
via their joint type. Moreover, if~$\bfx \in T_Q$ and $\bfy \in T_V(\bfx)$, then 
\begin{equation}
W^n(\bfy|\bfx) = e^{-n(D(V||W|Q)+H(V|Q))}.
\end{equation}
Thus, by introducing for every type $Q$, every
conditional type~$V$, and every $m \in \{1,\ldots,e^{nR}\}$ the RV
\begin{equation}
\tilde{B}_m(Q,V) = 1\bigl\{W^n(\bfy_{QV}|\bfX_m) \geq e^{-n(D(V||W|Q)+H(V|Q))}
\bigr\},
\end{equation}
where $\bfy_{QV}$ is an arbitrary sequence in $\set{Y}^n$ of type $QV$, we can rewrite~\eqref{eq:1894} as
\begin{equation}
\label{eq:1675}
\sum_{Q,V} (P\circ W)^n( T_{Q\circ V}) \BiggE{\biggl(1+\sum_{m=2}^{e^{nR}}
\tilde{B}_m(Q,V)\biggr)^\rho},
\end{equation}
where the sum extends over all types $Q$ and all conditional types $V$, and
where $P\circ W$ denotes the distribution on $\set{X}\times\set{Y}$ induced by $P$ and $W$
\begin{equation}
(P\circ W)(x,y) = P(x)W(y|x),\quad x\in\set{X},\,y\in\set{Y}.
\end{equation}
Next, we derive an upper-bound on~\eqref{eq:1675}. To this end,
note that for fixed $Q$ and $V$ the RVs 
\begin{equation}
\tilde{B}_1(Q,V),\ldots,\tilde{B}_{e^{nR}}(Q,V) 
\end{equation}
are IID Bernoulli. We can upper-bound their probability of success as follows.
\begin{align}
&\Pr(\tilde{B}_m(Q,V)=1)\notag\\
&=\Pr\bigl(W^n(\bfy_{QV}|\bfX_m) \geq e^{-n(D(V||W|Q)+H(V|Q))}\bigr)\notag\\
&=\Pr\bigl(W^n(\bfy_{QV}|\bfX_m)^{\frac{1}{1+\rho}} \geq
e^{-\frac{n}{1+\rho}(D(V||W|Q)+H(V|Q))}\bigr)\notag\\
&\leq e^{\frac{n}{1+\rho}(D(V||W|Q)+H(V|Q))}
\bigE{W^n(\bfy_{QV}|\bfX_m)^\frac{1}{1+\rho}},\label{eq:1687}
\end{align}
where~\eqref{eq:1687} follows from Markov's inequality.
As to the expectation on the right-hand side of~\eqref{eq:1687},
\begin{align}
&\bigE{W^n(\bfy_{QV}|\bfX_m)^\frac{1}{1+\rho}}\notag\\
&\quad= \prod_{i=1}^n
\bigE{W(y_{QV,i}|X_{m,i})^\frac{1}{1+\rho}}\label{eq:2522_independence}\\
&\quad=\prod_{y\in\set{Y}}\biggl(\sum_{x\in\set{X}} P(x) W(y|x)^{\frac{1}{1+\rho}}\biggr)^{n
(QV)(y)}\notag\\
&\quad=e^{n\sum_{y\in\set{Y}} (QV)(y) \log \sum_{x\in\set{X}} P(x)
W(y|x)^{\frac{1}{1+\rho}}}\notag\\
&\quad=e^{-n K(QV)},\label{eq:1697}
\end{align}
where~\eqref{eq:2522_independence} follows from the independence of the components of the
codewords, and where we define for every PMF $\tilde{P}$ on~$\set{Y}$
\begin{equation}
K(\tilde{P}) = -\sum_{y\in\set{Y}} \tilde{P}(y) \log \sum_{x\in\set{X}} P(x) W(y|x)^{\frac{1}{1+\rho}}.
\end{equation}
Substituting~\eqref{eq:1697} into~\eqref{eq:1687},
\begin{equation}
\Pr(\tilde{B}_m(Q,V)=1)\leq e^{-n\bigl(K(QV) -
\frac{D(V||W|Q)+H(V|Q)}{1+\rho}\bigr)}.
\end{equation}
Having bounded the probability of success of $\tilde{B}(Q,V)$, we next use
Lemma~\ref{lem:binomial} (Appendix~\ref{sec:binomial}), 
specifically~\eqref{eq:1547},  to conclude that the $\rho$-th 
moment in~\eqref{eq:1675} is bounded by 
\begin{equation}
\label{eq:1713_a}
1+\gamma e^{n(R-K(QV)+\frac{D(V||W|Q)+H(V|Q)}{1+\rho})}
\end{equation}
if $(Q,V) \in \set{G}(R)$, where 
\begin{multline*}
\set{G}(R)\\
= \biggl\{(Q,V): K(QV) - \frac{D(V||W|Q)+H(V|Q)}{1+\rho} \geq
R\biggr\},
\end{multline*}
and otherwise is bounded by 
\begin{equation}
\label{eq:1713_b}
\gamma e^{n\rho(R-K(QV)+\frac{D(V||W|Q)+H(V|Q)}{1+\rho})}.
\end{equation}
The other term in~\eqref{eq:1675} can be bounded as
\begin{align}
(P\circ W)^n(T_{Q\circ V}) &\leq e^{-nD(Q\circ V|| P\circ W)}\notag\\
&=e^{-n(D(Q||P)+D(V||W|Q))}.\label{eq:1718}
\end{align}
Using~\eqref{eq:1713_a},~\eqref{eq:1713_b} and~\eqref{eq:1718}, we can bound the summand
in~\eqref{eq:1675}. We treat separately the cases $(Q,V) \notin \set{G}(R)$ and $(Q,V)
\in \set{G}(R)$. In the former case,~\eqref{eq:1713_b} and~\eqref{eq:1718} give
\begin{multline}
\label{eq:2523_tube}
(P\circ W)^n(T_{Q\circ V})\BiggE{\biggl(1+\sum_{m=2}^{e^{nR}}
\tilde{B}_m(Q,V)\biggr)^\rho}\\ \leq \gamma e^{n\rho\bigl(R-K(QV)-\rho^{-1}D(Q||P) +
\frac{H(V|Q)-\rho^{-1}D(V||W|Q)}{1+\rho}\bigr)},\\
(Q,V) \notin \set{G}(R).
\end{multline}
We upper-bound the right-hand side of~\eqref{eq:2523_tube} in terms of $R$, $n$, $\rho$, and $E_0(\rho,P)$ by
showing that 
\begin{multline}
\label{eq:1963}
\min_{Q,V} \Bigl\{K(QV) + \rho^{-1} D(Q||P)\\
- \frac{H(V|Q)-\rho^{-1}
D(V||W|Q)}{1+\rho}\Bigr\} =  \frac{E_0(\rho,P)}{\rho},
\end{multline}
where the minimum is over all PMFs $Q$ on $\set{X}$ and all auxiliary channels
$V(y|x)$, $x\in \set{X}$, $y\in\set{Y}$. 
To establish~\eqref{eq:1963}, define 
\begin{equation}
\alpha(y) = \biggl(\sum_{x\in\set{X}} P(x) W(y|x)^{\frac{1}{1+\rho}}\biggr)^\rho,
\end{equation}
and observe that 
\begin{align}
& K(QV) + \rho^{-1} D(Q||P)- \frac{H(V|Q)-\rho^{-1}
D(V||W|Q)}{1+\rho}\notag\\
&=-\frac{1}{\rho}\sum_{y\in\set{Y}}\sum_{x\in\set{X}} Q(x) V(y|x) \log \frac{P(x)W(y|x)^{\frac{1}{1+\rho}}
\alpha(y)}{Q(x)V(y|x)}\notag\\
&\geq - \frac{1}{\rho}\log \sum_{y\in\set{Y}}\sum_{x\in\set{X}} P(x) W(y|x)^{\frac{1}{1+\rho}}
\alpha(y)\label{eq:2261_4342}\\
&=-\frac{1}{\rho}\log \sum_{y\in\set{Y}} \biggl(\sum_{x\in\set{X}} P(x)
W(y|x)^{\frac{1}{1+\rho}}\biggr)^{1+\rho}\notag\\
&=\frac{E_0(\rho,P)}{\rho},
\end{align}
where~\eqref{eq:2261_4342} follows from Jensen's Inequality.
The proof of~\eqref{eq:1963} is completed by noting that the choice 
\begin{equation*}
Q(x)V(y|x) = \frac{P(x)W(y|x)^{\frac{1}{1+\rho}}
\alpha(y)}{\sum_{x'\in\set{X},y'\in\set{Y}}P(x')W(y'|x')^{\frac{1}{1+\rho}}\alpha(y')}
\end{equation*}
achieves equality in~\eqref{eq:2261_4342}. 

Combining~\eqref{eq:2523_tube} with~\eqref{eq:1963} shows that
\begin{multline}
\label{eq:2570_ineq}
(P\circ W)^n(T_{Q\circ V})\BiggE{\biggl(1+\sum_{m=2}^{e^{nR}}
\tilde{B}_m(Q,V)\biggr)^\rho}\\
\leq \gamma e^{n\rho\bigl(R-\frac{E_0(\rho,P)}{\rho}\bigr)},\quad
(Q,V) \notin \set{G}(R),\; \rho>0.
\end{multline}
We now turn to the case where $(Q,V) \in \set{G}(R)$. 
We treat separately the subcases $\rho \geq 1$ and $0<\rho<1$, beginning with the
former. From~\eqref{eq:1718} and the fact that relative entropies are
nonnegative, it follows that
\begin{multline}
\label{eq:2578_blub}
(P\circ W)^n(T_{Q\circ V})\\ \leq e^{-n \rho^{-1}(D(Q||P)+D(V||W|Q))},\quad
\rho\geq 1.
\end{multline}
Combining~\eqref{eq:2578_blub} with~\eqref{eq:1713_a} and~\eqref{eq:1963} gives
\begin{multline}
\label{eq:2581_nice}
(P\circ W)^n(T_{Q\circ V})\BiggE{\biggl(1+\sum_{m=2}^{e^{nR}}
\tilde{B}_m(Q,V)\biggr)^\rho}\\ \leq (P\circ W)^n(T_{Q\circ V})+\gamma
e^{n\bigl(R-\frac{E_0(\rho,P)}{\rho}\bigr)},\\
(Q,V) \in \set{G}(R),\; \rho \geq 1.
\end{multline}
It remains to treat the case where $(Q,V) \in \set{G}(R)$ and $0<\rho<1$. In this
case, 
\begin{multline}
\label{eq:2561_lunch}
R-K(QV)+\frac{D(V||W|Q)+H(V|Q)}{1+\rho}\\
\leq \rho
\Bigl(R-K(QV)+\frac{D(V||W|Q)+H(V|Q)}{1+\rho}\Bigr),\\
(Q,V) \in \set{G}(R),\; 0<\rho < 1.
\end{multline}
Using~\eqref{eq:2561_lunch} to upper-bound the right-hand side
of~\eqref{eq:1713_a}, we obtain
\begin{multline}
\label{eq:2568_hair}
\BiggE{\biggl(1+\sum_{m=2}^{e^{nR}}
\tilde{B}_m(Q,V)\biggr)^\rho}\\ \leq 1+\gamma
e^{n\rho\bigl(R-K(QV)+\frac{D(V||W|Q)+H(V|Q)}{1+\rho}\bigr)},\\
(Q,V) \in \set{G}(R),\; 0<\rho<1.
\end{multline}
Combining~\eqref{eq:2568_hair} with~\eqref{eq:1718} and~\eqref{eq:1963} yields
\begin{multline}
\label{eq:3227_rev}
(P\circ W)^n(T_{Q\circ V})\BiggE{\biggl(1+\sum_{m=2}^{e^{nR}}
\tilde{B}_m(Q,V)\biggr)^\rho}\\ \leq (P\circ W)^n(T_{Q\circ V})+\gamma
e^{n\rho\bigl(R-\frac{E_0(\rho,P)}{\rho}\bigr)},\\
(Q,V) \in \set{G}(R),\; 0<\rho<1.
\end{multline}
Combining~\eqref{eq:2581_nice} with~\eqref{eq:2570_ineq} and using the fact that
the number of types and conditional types is polynomial in $n$, we obtain
\begin{multline}
\label{eq:3226_final}
\sum_{Q,V} (P\circ W)^n(T_{Q\circ V})\BiggE{\biggl(1+\sum_{m=2}^{e^{nR}}
\tilde{B}_m(Q,V)\biggr)^\rho}\\ \leq 1+ 
e^{n\bigl(R-\frac{E_0(\rho,P)}{\rho}+\delta_n\bigr)}+ 
e^{n\rho\bigl(R-\frac{E_0(\rho,P)}{\rho}+\delta_n\bigr)},\\
\rho\geq 1.
\end{multline}
Similarly, combining~\eqref{eq:3227_rev} with~\eqref{eq:2570_ineq}, we obtain
\begin{multline}
\label{eq:3235_final}
\sum_{Q,V} (P\circ W)^n(T_{Q\circ V})\BiggE{\biggl(1+\sum_{m=2}^{e^{nR}}
\tilde{B}_m(Q,V)\biggr)^\rho}\\ 
\leq 1+ e^{n\rho\bigl(R-\frac{E_0(\rho,P)}{\rho}+\delta_n\bigr)},\quad 0<\rho<1. 
\end{multline}
This completes the proof of~\eqref{eq:gallager_direct} 
because the right-hand sides of~\eqref{eq:3226_final} and~\eqref{eq:3235_final} 
tend to one as $n$ tends to infinity provided that
$R< E_0(\rho,P)/\rho$,
and we may choose a $P$ that maximizes the right-hand side. 
\qed

\section{A Proof that~\eqref{eq:const_comp_lb_cal} is at Least as Tight
as~\eqref{eq:forney_cal_rho}}
\label{appendix:D}
As pointed out in~\cite{telatar1997zero}, we may add the constraint $V'\ll
W$ in the minimization in~\eqref{eq:const_comp_lb_cal} without increasing the value of
the minimum. For any input PMF $P$ and any two auxiliary
channels $V,V'\ll W$ satisfying $PV=PV'$, 
\begin{align}
&-\rho^{-1} \log \sum_{y\in\set{Y}} (PW)(y) P(\set{X}(y))^\rho\notag\\
&\leq -\rho^{-1} \log \sum_{y\in\supp(PV')} (PW)(y)
P(\set{X}(y))^\rho\label{eq:2511_subset}\\
&= -\rho^{-1} \log \sum_{y\in\supp(PV')} (PV')(y) \frac{(PW)(y)
P(\set{X}(y))^\rho}{(PV')(y)}\label{eq:2514_algebraic}\\
&\leq -\rho^{-1}  \sum_{y\in\supp(PV')} (PV')(y) \log \frac{(PW)(y)
P(\set{X}(y))^\rho}{(PV')(y)}\label{eq:2516_jensen}\\
&=\rho^{-1} D(PV'||PW)+ \sum_{y\in\supp(PV')} (PV')(y) \log
\frac{1}{P(\set{X}(y))},\label{eq:2084}
\end{align}
where~\eqref{eq:2511_subset} follows because the support
of $PV'$ is a subset of the support of $PW$ (because $V'\ll W$); where~\eqref{eq:2514_algebraic}
follows by multiplying and dividing the summand by $(PV')(y)$;
and where~\eqref{eq:2516_jensen} follows from Jensen's Inequality. By the Log-Sum Inequality
\cite[Lemma~3.1]{csiszar2011information}
\begin{equation}
\label{eq:2091}
D(PV'||PW) \leq D(V'||W|P).
\end{equation}
The second term on the right-hand side of~\eqref{eq:2084} can be 
upper-bounded as follows.
\begin{align}
&\sum_{y\in\supp(PV')} (PV')(y) \log \frac{1}{P(\set{X}(y))}\notag\\
&=\sum_{y\in\supp(PV)} (PV)(y) \log
\frac{1}{P(\set{X}(y))}\label{eq:2533_blab}\\
&=\sum_{y\in\supp(PV)} (PV)(y) \log
\frac{(PV)(y)}{(PV)(y)P(\set{X}(y))}\notag\\
&\leq \sum_{y\in\supp(PV)} (PV)(y) \log
\frac{(PV)(y)}{(PV)(y)\sum_{x:V(y|x)>0}P(x)}\label{eq:2538_ineq}\\
&\leq \sum_{y\in\supp(PV)} \sum_{x: V(y|x)>0} P(x) V(y|x) \log
\frac{V(y|x)}{(PV)(y)}\label{eq:2540_logsum}\\
&=I(P,V),\label{eq:2103}
\end{align}
where~\eqref{eq:2533_blab} follows because $PV=PV'$; 
where~\eqref{eq:2538_ineq} follows because $V\ll W$; and
where~\eqref{eq:2540_logsum}
follows from the Log-Sum Inequality. 
Combining~\eqref{eq:2084} with~\eqref{eq:2091} and~\eqref{eq:2103} shows
that the right-hand side of~\eqref{eq:forney_cal_rho} never exceeds the right-hand side of~\eqref{eq:const_comp_lb_cal}.\qed

\section{A Property of Gallager's $E_0$ Function}
\label{appendix:gallager}
Gallager \cite{gallager1968information} defined the function
\begin{equation}
\label{eq:gallager_E0}
E_0(\rho,P) = -\log \sum_{y\in \set{Y}} \biggl( \sum_{x\in\set{X}}
P(x)W(y|x)^{\frac{1}{1+\rho}} \biggr)^{1+\rho},
\end{equation}
for all $\rho\geq 0$ and all PMFs $P$ on $\set{X}$. Here we show that
\begin{equation}
\label{eq:E0_limit}
\lim_{\rho\to\infty}
\max_P \frac{E_0(\rho,P)}{\rho}=-\log \pi_0,
\end{equation}
where $\pi_0$ is defined in~\eqref{eq:pi_0}.
This identity is noted without proof in~\cite{gallager1968information}.
To establish~\eqref{eq:E0_limit}, we first show that for any $P$
\begin{equation}
\label{eq:2279_0343}
\lim_{\rho\to\infty} \frac{E_0(\rho,P)}{\rho} = - \log \max_{y\in\set{Y}}
P(\set{X}(y)).
\end{equation}
We then use Lemma~\ref{lem:minimax} (Appendix~\ref{appendix:minimax}) to justify
the interchange of limit and maximization. The lemma applies
because $E_0(\rho,P)/\rho$ is nonincreasing and continuous in $\rho>0$ and continuous on the
set of all PMFs on $\set{X}$ (a compact subset of~$\Reals^{\card{\set{X}}}$).

To prove~\eqref{eq:2279_0343} for a given $P$, we distinguish two cases: 
Assume first that there exists $y_0\in\set{Y}$ such
that~$W(y_0|x)>0$ for all~$x\in\set{X}$ with~$P(x)>0$. In this case, the right-hand side
of~\eqref{eq:2279_0343} is equal to zero because $P(\set{X}(y_0)) = 1$. 
As to the left-hand side of~\eqref{eq:2279_0343}, note that replacing the sum over all
$y\in\set{Y}$ on the right-hand side
of~\eqref{eq:gallager_E0} with the term
corresponding to $y_0$ shows that
\begin{align}
\label{eq:2766_gallager}
E_0(\rho,P) \leq - (1+\rho) \log \sum_{x\in\set{X}}
P(x)W(y_0|x)^{\frac{1}{1+\rho}}.
\end{align}
Using L'Hospital's Rule, 
\begin{align}
&\lim_{\rho\to\infty} (1+\rho) \log \sum_{x\in\set{X}}
P(x)W(y_0|x)^{\frac{1}{1+\rho}}\notag\\
&= \lim_{\xi \searrow 0} \frac{\log \sum_{x\in
\set{X}} P(x)W(y_0|x)^\xi}{\xi}\notag\\
&=\sum_{x\in\set{X}} P(x) \log W(y_0|x). \label{eq:2587_lhopital}
\end{align}
Combining~\eqref{eq:2587_lhopital} and~\eqref{eq:2766_gallager},
\begin{equation}
\label{eq:2779_gallager}
\lim_{\rho\to\infty} E_0(\rho,P) \leq -\sum_{x\in\set{X}} P(x)\log W(y_0|x).
\end{equation}
Since the right-hand side of~\eqref{eq:2779_gallager} is a finite number, and
$E_0(\rho,P)\geq 0$, it follows that
\begin{equation}
\lim_{\rho\to \infty} \frac{E_0(\rho,P)}{\rho} =0.
\end{equation}
This establishes~\eqref{eq:2279_0343} for the first case. 
It remains to check the case where for every $y\in \set{Y}$ there is some $x_y\in \set{X}$
for which $P(x_y)>0$ and $W(y|x_y)=0$. 
In this case, for every $y\in \set{Y}$,
\begin{align}
\biggl( \sum_{x\in\set{X}}
P(x)W(y|x)^{\frac{1}{1+\rho}} \biggr)^{1+\rho} &\leq
\bigl(1-P(x_y)\bigr)^{1+\rho}\notag\\
&\to 0,\quad (\rho \to \infty).
\end{align}
Consequently, $E_0(\rho,P) \to \infty$ as $\rho \to \infty$, so by L'Hospital's Rule
\begin{equation}
\lim_{\rho\to\infty} \frac{E_0(\rho,P)}{\rho} = \lim_{\rho\to\infty}
\frac{\partial E_0(\rho,P)}{\partial \rho}.
\end{equation}
Straightforward computations show that
\begin{align}
\label{eq:2221_865}
&\frac{\partial E_0(\rho,P)}{\partial \rho}\notag= -\sum_{y\in\set{Y}}\frac{
\bigl(\sum_{x\in\set{X}}
P(x)W(y|x)^{\frac{1}{1+\rho}}\bigr)^{1+\rho}
}{\sum_{y'\in\set{Y}}\bigl(\sum_{x'\in\set{X}}P(x')W(y'|x')^{\frac{1}{1+\rho}}\bigr)^{1+\rho}}\notag\\
&\quad\quad\times\Bigl(\epsilon(\rho)+\log
\sum_{x''\in\set{X}}
P(x'')W(y|x'')^{\frac{1}{1+\rho}}
\Bigr),
\end{align}
where $\epsilon(\rho)\to0$ as $\rho\to\infty$. For each $y\in\set{Y}$, the
expression
\begin{multline}
\label{eq:2813_decay}
\biggl(\sum_{x\in\set{X}} P(x)W(y|x)^{\frac{1}{1+\rho}}\biggr)^{1+\rho}\\
=e^{(1+\rho) \log \sum_{x\in\set{X}}P(x)W(y|x)^{\frac{1}{1+\rho}}}
\end{multline}
is either zero for all $\rho>0$ or decays exponentially with $\rho$. Noting that
\begin{equation}
\lim_{\rho\to\infty}\sum_{x\in\set{X}}P(x)W(y|x)^{\frac{1}{1+\rho}} = P\bigl(\set{X}(y)\bigr),
\end{equation}
we see that the slowest decay in~\eqref{eq:2813_decay} occurs for those $y\in\set{Y}$ that maximize
$P(\set{X}(y))$. This implies that the right-hand side of~\eqref{eq:2221_865} approaches the right-hand side
of~\eqref{eq:2279_0343} as $\rho$ tends to infinity.\qed

\section{A Minimax Lemma}
\label{appendix:minimax}
\begin{lemma}
\label{lem:minimax}
Let $\set{C}$ be a compact subset of $\Reals^n$, let $\set{I}=[\alpha,\infty)$ for
some $\alpha\in\Reals$, and let $f\colon
\set{I}\times \set{C}
\to \Reals$ be such that $f(\cdot,\pi)$ is nonincreasing and continuous for every $\pi\in \set{C}$
and $f(\rho,\cdot)$ is
continuous for every $\rho\in\set{I}$. Then
\begin{equation}
\label{eq:2459_dfsdf}
\lim_{\rho \to \infty} \max_{\pi \in \set{C}} f(\rho,\pi) = \max_{\pi\in
\set{C}}\lim_{\rho\to\infty}
f(\rho,\pi).
\end{equation}
\end{lemma}
\begin{proof}
We first show that the maximum on the right-hand side of~\eqref{eq:2459_dfsdf} is attained. Select a sequence
$\pi_1,\pi_2,\ldots$ in $\set{C}$ such that 
\begin{equation}
\lim_{n\to\infty} \lim_{\rho\to\infty} f(\rho,\pi_n) = \sup_{\pi\in \set{C}}
\lim_{\rho\to\infty} f(\rho,\pi). 
\end{equation}
By compactness of $\set{C}$, we can find a
convergent subsequence $\pi_{n_k} \to \pi_\infty \in \set{C}$ as $k\to\infty$. 
By continuity and monotonicity we have for every $\rho_0\in\set{I}$ that
\begin{align}
f(\rho_0,\pi_\infty) &= \lim_{k\to\infty} f(\rho_0,\pi_{n_k})\notag\\
&\geq \lim_{k\to\infty} \lim_{\rho\to\infty} f(\rho,\pi_{n_k})\notag\\
&= \sup_{\pi\in \set{C}} \lim_{\rho\to\infty} f(\rho,\pi).
\end{align}
Taking $\rho_0\to\infty$ thus shows that $\pi_\infty$ attains the maximum on the right-hand side of~\eqref{eq:2459_dfsdf}.

To prove that equality holds in~\eqref{eq:2459_dfsdf}, first note that the left-hand side
is clearly never smaller than the right-hand side, so it remains to prove the
reverse inequality. If the left-hand side equals $-\infty$, then there is
nothing left to prove. 
Otherwise select real numbers $a$ and $b$ such that 
\begin{equation}
a<b<\lim_{\rho\to\infty} \max_{\pi\in \set{C}} f(\rho,\pi)
\end{equation}
and define the sets
\begin{subequations}
\begin{align}
\set{A}(\pi)&=\{\rho\in\set{I}: f(\rho,\pi) \leq a\},\\
\set{B}(\pi)&=\{\rho\in\set{I}: f(\rho,\pi) \leq b\}.
\end{align}
\end{subequations}
Our choice of $a$ and $b$ implies that $\set{A}(\pi) \subseteq \set{B}(\pi)$ and
$\bigcap_{\pi\in \set{C}} \set{B}(\pi)=\emptyset$. For a fixed $\pi\in \set{C}$, the
set $\set{B}(\pi)$ is either empty
or, by monotonicity and continuity, 
an interval of the form $[\lambda, \infty)$. 
If $\set{B}(\pi_0)=\emptyset$ for some $\pi_0\in \set{C}$, then
$f(\rho,\pi_0) > b$ for every $\rho\in\set{I}$, so $\lim_{\rho\to\infty}
f(\rho,\pi_0) \geq b$, and
hence $\max_{\pi\in \set{C}} \lim_{\rho\to\infty} f(\rho,\pi) \geq b>a$. If
$\set{B}(\pi) \neq \emptyset$
for every $\pi \in \set{C}$, then, since $\bigcap_{\pi\in \set{C}} \set{B}(\pi)=\emptyset$,
we can find a sequence $\pi_1,\pi_2,\ldots$ in $\set{C}$ such
that $\set{B}(\pi_n) = [\lambda_n,\infty)$ where $\lambda_n \to \infty$ as
$n\to\infty$. By compactness of $\set{C}$, we can then find a convergent 
subsequence $\pi_{n_k} \to \pi_\infty\in \set{C}$ as
$k\to\infty$. We claim that $\set{A}(\pi_\infty)=\emptyset$. Indeed, for if $\rho_0
\in \set{A}(\pi_\infty)$, i.e., if $f(\rho_0,\pi_\infty) \leq
a$, then by continuity $f(\rho_0, \pi_{n_k}) \leq b$ for all sufficiently large $k$,
i.e., $\rho_0 \in \set{B}(\pi_{n_k})$ for all sufficiently large $k$.
This leads to a contradiction because $\set{B}(\pi_{n_k}) = [\lambda_{n_k},\infty)$
and $\lambda_{n_k} \to \infty$ as $k\to\infty$ so
$\lambda_{n_k} > \rho_0$ for sufficiently large $k$. Thus, $\set{A}(\pi_\infty)= \emptyset$
and hence $\lim_{\rho\to\infty} f(\rho,\pi_\infty) \geq a$, so $\max_{\pi\in \set{C}}
\lim_{\rho\to\infty} f(\rho,\pi) \geq a$. Letting $a \nearrow \lim_{\rho\to\infty}
\max_{\pi\in \set{C}} f(\rho,\pi)$ completes the proof. 
\end{proof}

\section*{Acknowledgment}
We thank the anonymous reviewers for their helpful comments. 


\ifCLASSOPTIONcaptionsoff
  \newpage
\fi



\bibliographystyle{IEEEtran}
\bibliography{IEEEabrv,list_jrnl}
\end{document}